\documentclass[sn-mathphys]{sn-jnl}
\usepackage{amsthm} 
\usepackage{amsmath} 
\usepackage{amssymb}
\usepackage{physics}
\usepackage{textcomp}
\usepackage{xfrac}
\usepackage[shortlabels]{enumitem}

\theoremstyle{thmstyleone}%
\theoremstyle{thmstyletwo}%
\newtheorem{Example}{Example}%
\newtheorem{Remark}{Remark}%

\theoremstyle{thmstyleone}%
\newtheorem{Theorem}{Theorem}[section]
\newtheorem{Corollary}[Theorem]{Corollary}
\newtheorem{definition}[Theorem]{Definition}
\newtheorem{Lemma}[Theorem]{Lemma}
\newtheorem{Proposition}[Theorem]{Proposition}

\DeclareMathOperator{\Ker}{Ker}
\DeclareMathOperator{\Attr}{Attr}
\DeclareMathOperator\supp{supp}
\DeclareMathOperator\Fix{Fix}
\DeclareMathOperator{\Ran}{Ran}
\DeclareMathOperator{\spec}{spec}

\def\Nstar{\mathcal{N}_{\star}}

\begin{document}

\title[Decoherence-free algebras]{\center{Decoherence-free algebras in} \\quantum dynamics}

\author*[1,2]{\fnm{Daniele} \sur{Amato}}\email{daniele.amato@uniba.it}

\author[1,2]{\fnm{Paolo} \sur{Facchi}}\email{paolo.facchi@ba.infn.it}

\author[1,2,3]{\fnm{Arturo} \sur{Konderak}}\email{akonderak@cft.edu.pl}

\affil[1]{\orgdiv{Dipartimento di Fisica}, \orgname{Universit\`a di Bari}, \orgaddress{\postcode{70126} \city{Bari}, \country{Italy}}}

\affil[2]{\orgdiv{INFN Sezione di Bari},  \orgaddress{\postcode{70126} \city{Bari}, \country{Italy}}}

\affil[3]{\orgdiv{Center for Quantum-Enabled Computing, Center for Theoretical Physics, Polish Academy of Sciences},  \orgaddress{\street{Aleja Lotnik\'{o}w 32/46,} \postcode{02-688} \city{Warsaw}, \country{Poland}}}



\abstract{In this Article we analyze the algebraic properties of the asymptotic dynamics  of finite-dimensional open quantum systems in the Heisenberg picture. In particular, a natural product (Choi-Effros product) can be defined in the asymptotic regime. Motivated by this structure, we introduce a new space called the Choi-Effros decoherence-free algebra. Interestingly, this space is both a \textit{C*}-algebra with respect to the composition product, and a \textit{B*}-algebra with respect to the Choi-Effros product. Moreover, such space admits a direct-sum decomposition revealing a clear relationship with the attractor subspace of the dynamics. In particular, the equality between the attractor subspace and the Choi-Effros decoherence-free algebra is a necessary and sufficient condition for a faithful dynamics. Finally, we show how all the findings do not rely on complete positivity but on the much weaker Schwarz property.}

\maketitle

\section{Introduction}\label{intro}

Quantum technologies are evolving so rapidly over the last few years that a solid understanding of open-system theory, the theoretical framework of such applications, is crucially needed. In particular, asymptotics of open quantum systems has attracted a lot of interest by theoretical physicists~\cite{jex_st_2012,jex_st_2018,albert2019asymptotics}, in light of possible applications in quantum information processing~\cite{zanardi2014coherent}, quantum reservoir engineering~\cite{Zoller_res_eng,Wolf_res_eng}, and matrix product states~\cite{perez2006matrix}.
  
At the same time, long-time behavior of open quantum systems has drawn the attention of several mathematicians since the birth of the field in the early seventies~\cite{kossakowski1972quantum}. In particular, after the seminal works by Gorini, Kossakowski, Sudarshan~\cite{GKS_76}, and Lindblad~\cite{Lindblad_76} on quantum dynamical semigroups, several works~\cite{evans1977irreducible,Spohn_77,Frigerio_78,Frigerio_Verri_82} focused on the existence and the uniqueness of a stationary state of this kind of dynamics. 
The more recent papers by Fagnola and Rebolledo~\cite{fagnola_2001,fagnola2002subharmonic} can be framed in this research direction.

Motivated by the works on environmental decoherence by Blanchard and Olkiewicz~\cite{blanchard2003quantum,blanchard2003decoherence,blanchard2006decoherence} (see also~\cite{bulinskii1995some}), the asymptotic Heisenberg dynamics of open quantum systems was also studied in terms of its \textit{multiplicative} properties~\cite{dhahri2010decoherence,fagnola2008algebraic,fagnola2019role,carbone2020period}, using the notion of decoherence-free algebra~\cite{evans1977irreducible,Frigerio_Verri_82}. 

In particular, a decomposition of the dynamics into reversible and stable parts was sought~\cite{fagnola2019role,batkai2012decomposition}, following the seminal works by  Jacobs~\cite{jacobs1957fastperiodizitatseigenschaften}, DeLeeuw and Glicksberg~\cite{deleeuw1959almost,deleeuw1961applications}. Incidentally, multiplicative properties of the dynamics also turned out to have interesting connections with quantum error correction~\cite{choi2009multiplicative,rahaman2017multiplicative}. 
Also, multiplicative features of more general maps, including positive ones, were also taken into account in the literature~\cite{stormer2007multiplicative}.

With few exceptions~\cite{carbone2020period}, the majority of works adopting this approach study Markovian continuous-time evolutions. Moreover, the dynamics is usually supposed to be faithful.

A \textit{spectral} approach was instead adopted by Wolf and Perez-Garcia~\cite{wolf2010inverse}, Albert~\cite{albert2019asymptotics}, and the authors~\cite{AFK_asympt,AFK_asympt_2} in the Schr\"odinger picture and by Bhat \textit{et al.} \cite{rajarama2022peripheral} and again the authors~\cite{AFK_asympt_4} in the Heisenberg picture. This strategy defines the attractor subspace, where the asymptotic dynamics takes place, in terms of the peripheral (unimodular) eigenvalues of the evolution. Notice that those works tackled the problem of the asymptotics of open quantum systems \emph{without} the assumption of faithfulness or Markovianity. Moreover, Bhat~\textit{et al.}~\cite{rajaramaperipheral_2} and the authors~\cite{AFK_asympt_3} started to relate the multiplicative and spectral approaches for the asymptotic evolution.

In this Article we focus on the asymptotic discrete-time dynamics of a finite-dimensional open quantum system in the Heisenberg picture. 
Before summarizing the findings of the present work, let us recall the guiding principle behind the multiplicative strategy to the asymptotics addressed by the mathematical community over the last fifty years.

The seminal papers that introduce the decoherence-free algebra~\cite{blanchard2003quantum,blanchard2003decoherence,blanchard2006decoherence,bulinskii1995some} are based on the observation that a $\ast$-automorphism on the space of bounded operators on a Hilbert space is necessarily a unitary conjugation~\cite{stormer2012positive}. Since one expects that the large time-dynamics of an open quantum system is a unitary-like evolution, in the multiplicative perspective it is reasonable to identify the asymptotic subspace with the largest subspace on which the dynamics acts as a $\ast$-homomorphism, i.e. the decoherence-free algebra. Indeed, for Markovian evolutions, taken into account in~\cite{blanchard2003decoherence,bulinskii1995some}, the restriction of the dynamics to the decoherence-free algebra is a unitary conjugation, related to the Hamiltonian part of the generator of the evolution~\cite{dhahri2010decoherence}.

However, it is worthwhile to note that this is not generally true in the non-Markovian (even faithful) case as a consequence of a residual classical asymptotic dynamics~\cite{AFK_asympt,AFK_asympt_2}. This is in line with the fact that a $\ast$-automorphism on a ${C}^\ast$-algebra is not necessarily a unitary conjugation.

Moreover, another complication emerges in the non-faithful case: the composition product equipping the initial algebra of observables may not be appropriate in the asymptotic limit~\cite{contraction,contraction2} and, instead, the Choi-Effros product~\cite{choi1977injectivity} is more suitable. Therefore, we will study the multiplicative properties of a Heisenberg dynamics at large times with respect to this product.

In particular, we will introduce the Choi-Effros decoherence-free algebra, see Eq.~\eqref{star_dec_def}, which generalizes the standard one by replacing the composition product with the Choi-Effros product defined by Eq.~\eqref{star_prod_H}. Interestingly, it turns out that this space is a ${C}^\ast$-algebra with respect to the composition product, and a ${B}^\ast$-algebra with respect to the Choi-Effros product, see Theorems~\ref{sum_dec_N_th} and~\ref{C*-dec_th}, respectively. 


In particular, we will see that the Choi-Effros decoherence-free algebra admits a direct-sum decomposition (Theorem~\ref{sum_dec_N_th}), which naturally relates it to the attractor subspace defined by Eq.~\eqref{attr_def}, where the large-time evolution of the system occurs in the spectral perspective. Therefore, the Choi-Effros decoherence-free algebra always contains the attractor subspace [condition \textit{b)} of Theorem~\ref{mul_ch_UCP}] and, interestingly, equality characterizes faithful evolutions (Theorem~\ref{faith_mul_ch}). These properties reveal that, according to the algebraic structure of the attractor subspace (Theorem~\ref{C_star_UCP}), the Choi-Effros product is suitable for the study of the multiplicative properties of the dynamics in the asymptotic limit. Also, in the spirit of \cite{rajaramaperipheral_2}, these findings further bridge the gap between the spectral and multiplicative methods to address the asymptotics.

Finally, we will show in Section~\ref{Schw_CP} that all the results obtained in the present Article exploit a much weaker property than complete positivity, called the operator Schwarz inequality, see Eq.~\eqref{op_Schw_ineq}. Roughly speaking, the Schwarz property is sufficient for the description of the asymptotics of an open quantum system~\cite{Frigerio_Verri_82,wolf2010inverse}. Similarly, Schwarz positivity, rather than complete positivity, turns out to be sufficient to obtain a number of monotonicity results~\cite{carlen2022monotonicity,carlen2022characterizing} such as the data processing inequality~\cite{lindblad1975completely}, crucial in quantum information theory~\cite{nielsen2002quantum}. Notice also that complete positivity may be replaced by this property when we address the reversibility problem of quantum channels~\cite{jenvcova2012reversibility}.

The Article is organised as follows. Section~\ref{prel_sec} recalls some preliminary concepts needed in the following. In Section~\ref{star_dec_subsec} a new space related to Heisenberg asymptotics, called the Choi-Effros decoherence-free algebra, is introduced and its properties, as well as its relation with the attractor subspace, are discussed in detail. In Subsection~\ref{sec:qubit} we analyze the qubit case, while in Section~\ref{examples_sec} we show the application of the general results to several examples of qubit and qutrit UCP maps. Section~\ref{Schw_CP} shows that all the findings discussed in the previous sections for Heisenberg evolutions continue to be valid for the much larger class of Schwarz maps. After the conclusions (Section~\ref{concl}), we present in Appendix~\ref{Nstar_seminorm} an alternative way, based on seminorms, to recover the attractor subspace from the Choi-Effros decoherence-free algebra.
\section{Preliminaries}
\label{prel_sec}
In this section we set up the notation and summarize the main results on the algebraic properties of the asymptotic dynamics. 

%


First, remember that quantum mechanics can be described in different frameworks. According to the standard Dirac-Schr\"odinger picture, the fundamental mathematical objects of the theory are states, described as unit rays of a Hilbert space. Observables are then introduced as self-adjoint (linear) operators on such space. One of the more fundamental descriptions is the algebraic picture of quantum theory~\cite{John_God_Neumann_on_the,bratteli2012operator}. In this case, one starts with the algebra of observables, and then introduces states as linear functionals over it, detached from any Hilbert space structure~\cite{fgk}. Other descriptions of quantum mechanics could be considered, like the Schwinger picture~\cite{Schwinger_2018}, where quantum mechanics is encoded in the transitions~\cite{groupoidI,groupoidII,groupoidIII}.
		
		In this Article, we will mainly focus on the algebraic description, which turns out to be convenient for analyzing the mathematical structure of the asymptotics in the Heisenberg picture. In this picture, observables are defined as the self-adjoint elements of a $C^*$-algebra $\mathcal A$, namely a Banach space with a product $\cdot$, an involution $*$ and a norm $\norm{\ }$ satisfying the $C^\ast$-identity $\| a^\ast \!\cdot a \| = \| a \|^2, \, a \in \mathcal{A}$. We will frequently use the shorthand $( \mathcal{A} , \cdot , \| \cdot \| , \ast)$. 
%

		One of the main results in the theory of $C^\ast$-algebras, which motivated their definition, is the Gel'fand-Neumark theorem~\cite{gelfandNormierte,gelfand1943imbedding}. It shows that every $C^\ast$-algebra may be faithfully represented as a norm-closed $\ast$-subalgebra of bounded operators on some Hilbert space. Throughout this work, we will mostly consider the $C^\ast$-algebra $\mathcal{B}(\mathcal{H})$ of bounded operators on a finite-dimensional Hilbert space $\mathcal{H}$ and its $C^\ast$-subalgebras. From now on we will drop the symbol $\cdot$ in the equations for the composition product between bounded operators over $\mathcal{H}$. Moreover, in all the $\ast$-algebras we shall consider, the involution will be the adjoint in $\mathcal{B}(\mathcal{H})$, therefore we will simply denote a $B^\ast$- or $C^\ast$-algebra $\mathcal{A}$ as $( \mathcal{A} , \cdot , \| \cdot \|)$. 
		
		
		Once the observables have been defined, states are introduced as linear functionals over $\mathcal{A}$, $\omega: \mathcal A\rightarrow \mathbb C$ with $\omega(a^*a)\geqslant 0$, $a \in \mathcal{A}$, and $\norm{\omega}=1$. It is always possible to represent a state as a vector in a proper \textit{GNS} (or cyclic) representation $\pi_\omega$ on a Hilbert space
		$\mathcal H_\omega$ with $\omega(a)=\braket{\Omega_\omega}{\pi_\omega(a)\Omega_\omega}$.
		In particular, this representation is unique (up to a unitary equivalence). We will see in the following sections how this representation can be used to obtain (up to a $\ast$-automorphism) the attractor subspace as a suitable GNS Hilbert space of the Choi-Effros decoherence-free algebra.

	\subsection{Unital completely positive maps} 
	\label{asymptotic_definition_section}
		
		The present subsection is devoted to the basic properties of the dynamics in the Heisenberg picture, 
		which is natural for the algebraic picture of quantum theory, where observables are the primary objects. 
		 
		In the discrete-time description of the dynamics, the evolution in the Heisenberg picture of a finite-dimensional open quantum system is described by a unital completely positive (UCP) map $\Phi$ on $\mathcal{B}(\mathcal{H})$, where $\mathcal B(\mathcal H)$ is the algebra of bounded operators on a $d$-dimensional Hilbert space $\mathcal{H}$. Let $\braket{\cdot}{\cdot}$ be the scalar product in $\mathcal{H}$, which we will assume to be conjugate-linear in the first entry. 
		Here $\mathcal{B}(\mathcal{H})$ plays the role of the algebra of observables of the system undergoing the evolution. For a survey on positive and completely positive maps in the finite-dimensional setting see e.g.~\cite{chruscinski2022dynamical}. If the system observable at time $t=0$ is $A$, its evolution at time $t=n\in \mathbb{N}$ will be given by the action of the $n$-fold composition $\Phi^n$ of the map $\Phi$, i.e. $A(n)=\Phi^{n}(A)$.

		The spectrum $\spec(\Phi)$ of a UCP map $\Phi$, i.e. the set of its eigenvalues, satisfies the following three properties
		\begin{enumerate}[a)]
			\item $1 \in \spec(\Phi)$;
			\item $\lambda \in \spec(\Phi) \Rightarrow \bar{\lambda}\in \spec(\Phi)$;
			\item $\spec(\Phi)\subseteq \{ \lambda \in \mathbb{C} \,\vert \, \abs{\lambda} \leqslant 1 \}$.
		\end{enumerate}
		\begin{Remark}
			\label{spectr_UP}
			These properties still hold for the larger class of positive and unital maps~\cite{wolf2012quantum,Asorey2008}.
		\end{Remark}
		Following the spectral approach employed by the physics community \cite{wolf2010inverse,albert2019asymptotics,AFK_asympt,AFK_asympt_2}, the asymptotic dynamics is obtained in the large $n$-limit and it ends inside the asymptotic, peripheral or \textit{attractor subspace} of $\Phi$, defined as
		\begin{equation}
			\label{attr_def}
			\Attr(\Phi):= \mbox{span} \{ X \in \mathcal{B}(\mathcal{H}) \,,\, \Phi(X)=\lambda X \mbox{ for some } \lambda \in \spec_{\mathrm{P}}(\Phi) \} ,
		\end{equation}
		where 		
		\begin{equation}
			\spec_{\mathrm{P}}(\Phi):=\{ \lambda \in \spec(\Phi) \,,\, \abs{\lambda}=1 \}
		\end{equation}
		is the \textit{peripheral spectrum} of $\Phi$.
		
		A crucial property of the attractor subspace 
		of $\Phi$ which we will heavily use in the following is~\cite[Proposition 6.12]{wolf2012quantum}:
		\begin{equation}
			\label{Attr_inv_sub}
			\Phi\Attr(\Phi)=\Attr(\Phi),
		\end{equation}
		i.e. the attractor subspace $\Attr(\Phi)$ is invariant under $\Phi$. A relevant subspace of $\Attr(\Phi)$ is the \emph{fixed-point space} of $\Phi$:
		\begin{equation}
			\Fix(\Phi)= \{ X \in \mathcal{B}(\mathcal{H}) \,,\, \Phi(X)=X \},
		\end{equation}
		namely the eigenspace of $\Phi$ corresponding to the eigenvalue $\lambda = 1$. The space $\Fix(\Phi)$ contains the \textit{stationary observables} of the dynamics described by $\Phi$.
		
		Consider the spectral or \textit{Jordan} decomposition of a UCP map $\Phi$
		\begin{equation}
			\Phi=\sum_{\lambda_k \in \spec(\Phi)} (\lambda_k\mathcal{P}_k + \mathcal{N}_k),
			\label{Jordan_dec}
		\end{equation}
		where $\mathcal{P}_{k}$ and $\mathcal{N}_{k}$ are the eigenprojections and eigenilpotents of $\Phi$ corresponding to the $k$th eigenvalue $\lambda_{k}$~\cite{kato2013perturbation}. 
		Importantly, all peripheral eigenvalues $\lambda_k \in \spec_{\mathrm{P}}(\Phi)$ are semisimple, i.e. the corresponding eigennilpotents $\mathcal{N}_k$ vanish~\cite[Proposition 6.2]{wolf2012quantum}. Define the \textit{peripheral projection} $\mathcal{P}_{\mathrm{P}}$ of $\Phi$ as
		\begin{align}
			\label{P_p_def}
			\mathcal{P}_{\mathrm{P}}&=\sum_{\lambda_k\in\spec_{\mathrm{P}}(\Phi)}\mathcal{P}_k.
		\end{align}
		Note that it is a projection onto the attractor subspace $\Attr(\Phi)$ of $\Phi$. Also, the map $\mathcal{P}_{\mathrm{P}}$ can be written in terms of the UCP map $\Phi$ as 
			
			\begin{equation}
			\label{P_p_for}
			\mathcal{P}_{\mathrm{P}}=\lim_{i\rightarrow \infty}\Phi^{n_i},
			\end{equation}
		for some increasing sequence $( n_i)_{i\in\mathbb{N}}$. 
		From~\eqref{P_p_for} it is clear that $\mathcal{P}_{\mathrm{P}}$ is a UCP map, since the set of UCP maps is convex, compact, and closed under composition. In the subsequent sections, we will largely apply the following property to the peripheral projection $\mathcal{P}_{\mathrm{P}}$ of a UCP map $\Phi$~\cite[Lemma 2.4]{hamana1979injective}.
	\begin{Lemma}
			\label{lemma_Hamana}
			Let $\Phi$ be an idempotent UCP map, namely $\Phi^2=\Phi$. Then
			\begin{equation}
				\Phi(\Phi(X)\Phi(Y))=\Phi(\Phi(X)Y)=\Phi(X\Phi(Y)), \quad X,Y \in\mathcal{B}(\mathcal{H}).
			\end{equation} 
		\end{Lemma}
It is often useful to regard $\mathcal{B}(\mathcal{H})$ as a Hilbert space with respect to the \textit{Hilbert-Schmidt scalar product}
		\begin{equation}
			\braket{A}{B}_{\mathrm{HS}}= \tr(A^\ast B), \quad A,B \in \mathcal{B}(\mathcal{H}).
		\end{equation}
		The Schrödinger dynamics is given by $\Phi^\dagger$, the Hilbert-Schmidt adjoint of $\Phi$. In particular, $\Phi^\dagger$ is a completely positive trace-preserving (CPTP) map on $\mathcal{B}(\mathcal{H})$, called in the physics literature a \textit{quantum channel}~\cite{nielsen2002quantum}. Therefore, if $\rho$ is a state of the system, i.e. a density operator on $\mathcal{H}$, then $\Phi^\dagger(\rho)$ describes the corresponding evolved state. A state $\rho$ is said to be \textit{stationary} whenever it is invariant under the Schrödinger evolution
		\begin{equation}
			\Phi^\dagger(\rho)=\rho.
		\end{equation}
We can define the attractor subspace $\Attr(\Phi^\dagger)$ of the Schrödinger dynamics $\Phi^\dagger$ as in Eq.~\eqref{attr_def} as well as the peripheral projection, which turns out to be the adjoint $\mathcal{P}_{\mathrm{P}}^\dagger$ of $\mathcal{P}_{\mathrm{P}}$ by~\eqref{P_p_for}.
	
Within UCP maps, the \textit{faithful} ones are particularly relevant for physical reasons. They are also referred to as stationary maps~\cite{rajaramaperipheral_2}, and they are defined as follows.
		\begin{definition}
			A UCP map $\Phi$ is said to be \textit{faithful} if $\Phi^\dagger$, the Hilbert-Schmidt adjoint of $\Phi$, admits an invertible stationary state $\rho$, i.e.
			\begin{equation}
				\Phi^\dagger (\rho)=\rho >0.
			\end{equation} 
		\end{definition}		
		We will see in the following sections how faithful maps are very well-behaved in the asymptotic limit. 
		
		Before going on, observe that in the previous discussion we always assumed the map $\Phi$ to be UCP. In particular, all such maps satisfy the \textit{operator Schwarz inequality}:
		\begin{equation}
			\label{op_Schw_ineq}
			\Phi(X^\ast X) \geqslant \Phi(X)^\ast \Phi(X), \quad X \in \mathcal{B}(\mathcal{H}). 
		\end{equation}
				
		Importantly,~\eqref{op_Schw_ineq} is the only property used to prove all the results of the present work (see Section~\ref{Schw_CP} for a discussion on this point).
		\subsection{Algebraic structure of the asymptotics}
		\label{algebraic_structure_attractor_section}
		
		In this section we will summarize the main algebraic properties of the asymptotic dynamics.

 It is possible to introduce the structure of a $C^*$-algebra on the attractor subspace. This is done by means of the \textit{Choi-Effros product}, see~\cite[Theorem~3.1]{choi1977injectivity},~\cite[Proposition~5.3]{fidaleo2022spectral}, and~\cite[Theorem~2.3]{rajarama2022peripheral}. See also~\cite[Theorem 3]{AFK_asympt_4} for the explicit structure of the attractor subspace.
		\begin{Theorem}
			\label{C_star_UCP}
			Let $\Phi$ be a UCP map. Then its attractor subspace $\Attr(\Phi)$ is a unital ${C}^\ast$-algebra with respect to the Choi-Effros product
			\begin{equation}
				\label{star_prod_H}
				X \star Y := \mathcal{P}_{\mathrm{P}}(XY),	
			\end{equation}
		and the operator norm $\|\cdot \|$. We will denote this structure by $(\Attr(\Phi) , \star , \| \cdot \|)$.
		\end{Theorem}	
		Now, Eq.~\eqref{P_p_for} suggests that the Choi-Effros product is the suitable one in order to equip $\Attr(\Phi)$ with an algebraic structure, differently from the composition product $\cdot$.
		Observe that we can define the Choi-Effros product for all $X,Y \in \mathcal{B}(\mathcal{H})$, using Eq.~\eqref{star_prod_H}. In this way, we get a bilinear but non-associative product. For instance, take $d \geqslant 2$, and consider the idempotent UCP map on the space of $d$-dimensional matrices $\mathcal{M}_{d}$ defined as
		\begin{equation}
			\label{example_non_associative_choi}
			{\Phi}(X)=\frac{\tr(X)}{d}\mathbb{I}.
		\end{equation}
		Then, even for diagonal matrices $X,Y,Z$ the equality
		\begin{equation}
			X \star (Y \star Z) = (X \star Y) \star Z
		\end{equation}
		does not hold in general. In Section~\ref{star_dec_subsec} we will introduce a larger space $\Nstar$ in which the Choi-Effros product is associative, see Lemma~\ref{asso_N_CE_prop}.
		
		
%
		Following~\cite{AFK_asympt,AFK_asympt_2}, in which we consider the equivalent Schr\"odinger picture of the evolution, we can describe the asymptotic Heisenberg dynamics of an open quantum system by introducing the \textit{asymptotic map} $\Phi_{\mathrm{as}}$ of $\Phi$ as
		\begin{equation}
			\label{as_map_def}
			\Phi_{\mathrm{as}}= \Phi \vert_{\Attr(\Phi)}.
		\end{equation}
		Several general properties of the asymptotic map $\Phi_{\mathrm{as}} $ are collected in the next result, see again~\cite[Theorem 3]{AFK_asympt_4} for the explicit expression of $\Phi_{\mathrm{as}}$.
				\begin{Theorem}
			\label{aut_th_UCP}
			Let $\Phi$ be a UCP map and consider the $C^\ast$-algebra $(\Attr(\Phi) , \star , \| \cdot \| )$, where $\| \cdot \|$ denotes the operator norm on $\mathcal{B}(\mathcal{H})$.	
			
			Then the asymptotic map $\Phi_{\mathrm{as}}$ is a $\ast$-automorphism on $(\Attr(\Phi) , \star , \| \cdot \| )$,		
			 i.e. $\Phi_{\mathrm{as}}$ is invertible and
			\begin{equation}
				\label{aut_Phi_dagger}
				\Phi_{\mathrm{as}} (X \star Y) = \Phi_{\mathrm{as}}(X) \star \Phi_{\mathrm{as}}(Y),\quad X,Y\in \Attr(\Phi).
			\end{equation}
			Moreover, $\Phi_{\mathrm{as}}$ 
			is a UCP map on $(\Attr(\Phi) , \star , \| \cdot \|)$, and
\begin{equation}
				\label{iso_UCP}
				\| \Phi_{\mathrm{as}}(X) \| = \| X \|, \quad X \in \Attr(\Phi).
			\end{equation}
\end{Theorem}
		\begin{proof}
		The $\ast$-automorphic property of $\Phi_{\mathrm{as}}$ was proved in~\cite[Theorem 2.12]{rajarama2022peripheral}, while its complete positivity is a consequence of the first part of the claim and~\cite[Lemma 1.2.2 (ii)]{stormer2012positive}. Finally, $\Phi_{\mathrm{as}}$ is an isometry on~$(\Attr(\Phi) , \star , \| \cdot \|)$ 	by~\cite[Theorem 4.8]{conway2019course}. 
		\end{proof}

		The Choi-Effros product is in general different from the composition product. A UCP map will be called \textit{peripherally automorphic}~\cite{rajaramaperipheral_2} whenever the two products coincide within its attractor subspace.
		\begin{definition}
			\label{per_aut_def}
			Let $\Phi$ be a UCP map. Then $\Phi$ is said to be peripherally automorphic if and only if 
			\begin{equation}
				\label{per_aut_eq}
				X \star Y = XY, \quad X,Y \in \Attr(\Phi).
			\end{equation}
		\end{definition}
		Two necessary and sufficient conditions for this property to hold are provided in the next theorem~\cite[Theorem 2.10]{rajaramaperipheral_2}. In particular, condition \textit{c)}  explains the name given to this class of maps.
		
		\begin{Theorem}
			\label{ch_per_aut_1}
			Let $\Phi$ be a UCP map. The following conditions are equivalent:
			\begin{enumerate}[a)]
				\item $\Phi$ is peripherally automorphic;
				\item $\Attr(\Phi)$ is closed under the composition product;
				\item $\Phi(XY)=\Phi(X)\Phi(Y)$, $X,Y\in \Attr(\Phi)$.
			\end{enumerate}
		\end{Theorem} 

	Finally, it is worthwhile observing that the peripheral automorphism property~\eqref{per_aut_eq} is characterized by many mathematical properties, that are only necessary conditions for faithful UCP maps. See Theorem~\ref{mul_ch_pa_UCP}.

\subsection{Decoherence-free algebra}
\label{dec_free_known_results}
In this section we will recall the definition and several relevant properties of the decoherence-free algebra. Inspired by these properties, we will introduce in Section~\ref{star_dec_subsec} the Choi-Effros decoherence-free algebra. We will see that this algebra has many interesting mathematical features, and in particular a clear relationship with the attractor manifold.

Let us try to motivate the introduction of the decoherence-free algebra. It is well-known that a unitary UCP map on $\mathcal{B}(\mathcal{H})$, which describes the Heisenberg dynamics of a closed quantum system, is a $\ast$-automorphism and, indeed, the latter condition characterizes this class of maps, see e.g. \cite[Proposition 3.1.7]{stormer2012positive}.
Thus, since the asymptotic evolution should resemble a closed-system dynamics, the \textit{decoherence-free algebra} of a UCP map $\Phi$ is defined precisely as the subset of $\mathcal{B}(\mathcal{H})$ where $\Phi$ acts as a $\ast$-homomorphism, namely:

\begin{equation}
\begin{split}
\mathcal{N}:= \{ X \in \mathcal{B}(\mathcal{H}) \,,\, \Phi^{n}(Y X) &= \Phi^{n}(Y)\Phi^{n}(X),\, \\ \Phi^{n}(XY) &= \Phi^{n}(X)\Phi^{n}(Y), \; \forall Y \in \mathcal{B}(\mathcal{H}),\forall n \in \mathbb{N} \}.
\label{dec_def}
\end{split}
\end{equation}
It is well-known that $\mathcal{N}$ is a ${C}^\ast$-subalgebra of $\mathcal{B}(\mathcal{H})$, see e.g.~\cite{carbone2020period}.
Indeed, $\mathcal{N}$ is the largest subalgebra of $\mathcal{B}(\mathcal{H})$ on which $\Phi^n \vert_{\mathcal{N}}$ is a $\ast$-homomorphism. However, as already observed in~\cite[Remark 1]{carbone2020period}, $\Phi\vert_{\mathcal{N}}$ (and, consequently, $\Phi^n \vert_{\mathcal{N}}$) is \emph{not}  a $\ast$-automorphism in general.

The space $\mathcal{N}$ can be characterized in a simpler way by using~\cite[Theorem 5.4]{wolf2012quantum}.
\begin{Proposition}
\label{ch_N}
Let $\Phi$ be a UCP map. Then, its decoherence-free algebra is given by 
\begin{equation}
\begin{split}
	\mathcal{N}=\{ X \in \mathcal{B}(\mathcal{H}) \,,\, \Phi^{n}(X^\ast X)
	&= \Phi^{n}(X)^\ast\Phi^{n}(X),\\ \Phi^{n}(XX^\ast ) 
	&= \Phi^{n}(X)\Phi^{n}(X)^\ast, \,\forall n \in \mathbb{N} \}.
\end{split}
\end{equation}
\end{Proposition}
The following result elucidates the connection between faithfulness and peripheral automorphism, focusing on the inclusion relationship between $\mathcal N$ and $\Attr(\Phi)$ in the faithful case. The second part of the claim follows by combining Lemma 1 and Theorem 1 of~\cite{carbone2020period}, where the authors work in the infinite-dimensional setting. Nevertheless, we provide a simple proof of the result in the finite-dimensional case for completeness, see also~\cite[Proposition 8]{carbone2013Markov} for the Markovian case.
\begin{Theorem}
\label{attr_ch_faith}
Let $\Phi$ be a faithful UCP map with decoherence-free algebra $\mathcal{N}$. Then $\Phi$ is peripherally automorphic. Moreover, 
\begin{equation}
\label{eq_Attr_N_faith}
\Attr(\Phi)=\mathcal{N}.
\end{equation}
\end{Theorem} 
\begin{proof}
The first part of the assertion was already proved in~\cite[Corollary 2.8]{rajaramaperipheral_2}.

Let us now prove that $\Attr(\Phi) \subseteq \mathcal{N}$ in the faithful case.
Given $X\in \Attr(\Phi)$, from the first part of the claim and Theorem~\ref{ch_per_aut_1}
\begin{equation}
\Phi(X^\ast X)=\Phi(X)^\ast\Phi(X),
\end{equation}
and, since $\Attr(\Phi)$ is invariant under $\Phi$, see Eq.~\eqref{Attr_inv_sub}, we have:
\begin{equation}
	\Phi^n(X^\ast X)=\Phi^n(X)^\ast\Phi^n(X), \quad n\in\mathbb{N}.
\end{equation}
By exchanging $X \leftrightarrow X^\ast$, we get $\Attr(\Phi)\subseteq \mathcal N $.

For the opposite inclusion, let $X \in \mathcal{N}$. From~\eqref{P_p_for}
\begin{equation}
\Phi^n(X^\ast X)=\Phi^n(X)^\ast\Phi^n(X) \Rightarrow \mathcal{P}_{\mathrm{P}}(X^\ast X)=\mathcal{P}_{\mathrm{P}} (X)^\ast\mathcal{P}_{\mathrm{P}} (X).
\end{equation}
Thus, by applying Lemma~\ref{lemma_Hamana} to $\mathcal{P}_{\mathrm{P}}$ we obtain 
\begin{equation}
\mathcal{P}_{\mathrm{P}} ((X- \mathcal{P}_{\mathrm{P}}(X))^\ast (X-\mathcal{P}_{\mathrm{P}}(X)))=0.
\end{equation}
Let $Y = X- \mathcal{P}_{\mathrm{P}}(X)$, and let  $\rho$ be an  invertible invariant state for $\Phi^\dagger$. Then
\begin{equation}
\tr(\mathcal{P}_{\mathrm{P}}(Y^\ast Y) \rho)=\mathrm{tr}(Y^\ast Y \mathcal{P}_{\mathrm{P}}^\dagger(\rho))=\mathrm{tr}(Y^\ast Y \rho)=0,
\end{equation}
which implies $Y^\ast Y=0$, and so $Y=0$. Thus $X=\mathcal P_{\mathrm{P}}(X)$, and $\mathcal N\subseteq \Attr(\Phi)$. 
\end{proof}
\begin{Remark}
	It was already noted in~\cite[Remark 2.3, Example 2.4]{rajaramaperipheral_2}, a peripherally automorphic UCP map is not faithful in general. 
\end{Remark}
\begin{Remark}
	As a consequence of Theorem~\ref{aut_th_UCP}, in the faithful case $\Phi \vert_{\mathcal{N}}$ is a $\ast$-automorphism. The converse is not true, since any invertible (as a linear map) non-faithful UCP map {clearly satisfies this property too}, see e.g. Example~\ref{ex_2}.
\end{Remark}
Theorem~\ref{attr_ch_faith} can be generalized to peripherally automorphic UCP maps~\cite[Theorem 2.6]{rajaramaperipheral_2}.
\begin{Theorem}
\label{mul_ch_pa_UCP}
Let $\Phi$ be a UCP map and $\mathcal{N}$ its decoherence-free algebra. Then, the following conditions are equivalent.
\begin{enumerate}[a)] 
\item $\Phi$ is peripherally automorphic;
\item $\Attr(\Phi) \subseteq\mathcal{N}$;
\item $\Phi(X)=\lambda X$, $\lambda \in \spec_{\mathrm{P}}(\Phi)$ $\Rightarrow$ $\Phi(X^\ast  X)=X^\ast  X$.
\end{enumerate}
\end{Theorem}
In Example~\ref{ex_2} we are going to see that inclusion \textit{b)} can be strict. However, according to Theorem~\ref{attr_ch_faith}, it is saturated for faithful UCP maps.

\section{Choi-Effros decoherence-free algebra}\label{star_dec_subsec}

This section, in which we state the main results of our work, is based on the following crucial observation: as a consequence of the operator Schwarz inequality~\eqref{op_Schw_ineq}, any UCP map $\Phi$ satisfies 
\begin{equation}
\label{star_op_Schw_ineq}
\Phi(X^\ast \star X) \geqslant \Phi(X)^\ast \star \Phi(X), \quad X\in\mathcal{B}(\mathcal{H}),
\end{equation}
since its peripheral projection $\mathcal{P}_{\mathrm{P}}$ is a UCP map and commutes with $\Phi$ by~\eqref{P_p_for}. We will call~\eqref{star_op_Schw_ineq} the \textit{$\star$-operator Schwarz inequality}, because it is  clearly an analogue of~\eqref{op_Schw_ineq} with respect to the $\star$ product.
Motivated by the algebraic structure of $\Attr(\Phi)$, let us define the \textit{Choi-Effros decoherence-free algebra}~$\Nstar$ 
\begin{equation}
\begin{split}
\label{star_dec_def}
\Nstar:= \{ X \in \mathcal{B}(\mathcal{H})  \,,\, \Phi^{n}(Y \star X) &= \Phi^{n}(Y)\star \Phi^{n}(X), \\ \Phi^{n}(X \star Y) &= \Phi^{n}(X) \star \Phi^{n}(Y), \, \forall Y \in \mathcal{B}(\mathcal{H}),\forall n \in \mathbb{N} \}.
\end{split}
\end{equation}
The term ``algebra" attributed to this space will be justified precisely in Theorem~\ref{sum_dec_N_th}. $\Nstar$ may be regarded as the analogue of the attractor subspace in the multiplicative approach to the asymptotic evolution. Indeed, we will see in Theorem~\ref{sum_dec_N_th} that there is a clear inclusion relation between these two spaces. 
\begin{Remark}
	We introduced the Choi-Effros product in Theorem~\ref{C_star_UCP} for elements in the attractor subspace $\Attr(\Phi)$ of $\Phi$. Nevertheless, as discussed before Eq.~\eqref{example_non_associative_choi}, we could extend this definition to all elements $X,Y$ in $\mathcal B(\mathcal H)$, obtaining a non-associative product. It will turn out that this product is associative on $\Nstar$, see Lemma~\ref{asso_N_CE_prop}.
\end{Remark}
Let us start the analysis of $\Nstar$ with the following lemma, which can be proved along the lines of~\cite[Theorem 5.4]{wolf2012quantum}.
\begin{Lemma}
\label{lemma_lemme}
Let $\Phi$ be a UCP map. Then
\begin{align}
\label{impl}
&\Phi (X^\ast \star X) = \Phi (X)^\ast\star \Phi (X) \quad \Leftrightarrow \quad \Phi (Y \star X) = \Phi (Y)\star \Phi (X) , \quad \forall Y \in \mathcal{B}(\mathcal{H}),\nonumber\\
&\Phi (X \star X^\ast) = \Phi (X)\star \Phi (X)^\ast \quad \Leftrightarrow \quad \Phi (X \star Y) = \Phi (X)\star \Phi (Y) , \quad \forall Y \in \mathcal{B}(\mathcal{H}).
\end{align}
\end{Lemma}
The following result discusses several properties of the space $\Nstar$. In particular, it shows its relation with $\mathcal{N}$, and its invariance under $\Phi$, a property which is shared with $\mathcal N$.

\begin{Proposition}
\label{ch_N_star_th}
Let $\Phi$ be a UCP map. Then, its Choi-Effros decoherence-free algebra $\Nstar$ is given by
\begin{equation}
\begin{split}
\label{ch_N_star}
\Nstar=\{ X \in \mathcal{B}(\mathcal{H}) \,, \, \Phi^n (X^\ast \star X) &= \Phi^n (X)^\ast\star \Phi^n (X) ,\\ \Phi^n (X \star X^\ast ) &= \Phi^n (X)\star \Phi^n (X)^\ast, \,\forall n \in \mathbb{N} \}.
\end{split}
\end{equation}
Moreover, the following properties hold:
\begin{enumerate}[a)]
\item $\mathcal{N} \subseteq \Nstar$;
\item $\Phi(\Nstar) \subseteq \Nstar$.
\end{enumerate}
\end{Proposition}
\begin{proof}
Equation \eqref{ch_N_star} follows from~\eqref{star_dec_def} and Lemma~\ref{lemma_lemme}.
Then, property~\textit{a)} simply follows from
\begin{equation}
\Phi^n(X^\ast X) = \Phi^n(X)^\ast \Phi^n(X) \Rightarrow \Phi^n(X^\ast \star X) = \Phi^n(X)^\ast \star \Phi^n(X),
\end{equation}
for any $n \in \mathcal{N}$, which is a consequence of the commutativity between $\mathcal P_{\mathrm{P}}$ and~$\Phi$. 
Concerning property~\textit{b)}, given $X\in \Nstar$ and $n\in \mathbb{N}$, we have
\begin{equation}
\Phi^{n}(\Phi(X)^\ast \star \Phi(X))=\Phi^{n+1}(X^\ast \star X)=\Phi^{n+1}(X)^\ast \star \Phi^{n+1}(X),
\end{equation}
and $\Phi(X)\in\Nstar$ by~\eqref{ch_N_star}.
\end{proof}
The first part of Proposition~\ref{ch_N_star_th} characterizes $\Nstar$ as Proposition~\ref{ch_N} does to $\mathcal{N}$.
Indeed, we will go beyond this result, and we will see that we can simply characterize this space in terms of the peripheral projection $\mathcal P_{\mathrm{P}}$ of $\Phi$, see Corollary~\ref{N_star_in_N_star_P}. Again, this reveals the natural interplay between the space $\Nstar$ and $\Attr(\Phi)=\Ran(\mathcal{P}_{\mathrm{P}})$, the analogue of $\Nstar$ in the spectral approach to the asymptotics.
We are now going to show that the multiplicative properties characterizing peripherally automorphic UCP maps in Theorem~\ref{mul_ch_pa_UCP} hold in the general case when the composition product is replaced by the Choi-Effros product~\eqref{star_prod_H}.
\begin{Theorem}
Let $\Phi$ be a UCP map and $\Nstar$ its Choi-Effros decoherence-free algebra. Then the following conditions hold
\begin{enumerate}[a)] 
\item $\Phi(X \star Y)=\Phi(X) \star \Phi(Y)$, $\forall X,Y \in \Attr(\Phi)$;
\item $\Attr(\Phi) \subseteq\Nstar$;
\item $\Phi(X)=\lambda X$, $\lambda \in \spec_{\mathrm{P}}(\Phi)$ $\Rightarrow$ $\Phi(X^\ast \star X)=X^\ast \star X$.
\end{enumerate}
\label{mul_ch_UCP}
\end{Theorem}
\begin{proof} We use the notation from the theorem.
		
\textit{a)} It is just the assertion of Theorem~\ref{aut_th_UCP}.

\textit{b)} Given $X \in \Attr(\Phi)$ and repeating $n$ times condition \textit{a)}:
\begin{align}
&\Phi^n (X^\ast \star X) = \Phi^n(X)^\ast \star \Phi^n(X), \nonumber\\
&\Phi^n (X \star X^\ast) = \Phi^n(X) \star \Phi^n(X)^\ast.
\end{align} 
In particular, we used the invariance of the attractor subspace $\Attr(\Phi)$, see Eq.~\eqref{Attr_inv_sub}. 
%

%
\textit{c)} If $\Phi(X)=\lambda X$, with $\lambda \in \spec_{\mathrm{P}}(\Phi)$, then $X \in \Attr(\Phi)$, so the claim follows from \textit{a)}. 
\end{proof}
We are now going to study the algebraic properties of $\Nstar$ with respect to the composition product $\cdot$ of $\mathcal B(\mathcal H)$. To this purpose, we shall find a (vector space) direct-sum decomposition of $\Nstar$, see~\cite[Theorem 6]{AFK_asympt_4} for a similar result.
\begin{Theorem}[Structure  of $\Nstar$]
\label{sum_dec_N_th}
Let $\Phi$ be a UCP map with peripheral projection $\mathcal{P}_{\mathrm{P}}$. Then its Choi-Effros decoherence-free algebra $\Nstar$ admits the following direct-sum decomposition
\begin{equation}
\label{sum_dec_N_star}
\Nstar=\Attr(\Phi) \oplus \mathcal{K}({\mathcal{P}_{\mathrm{P}}}).
\end{equation}
Here, for a UCP map $\Psi$, we defined
\begin{equation}
\label{eq:definition_K_idempotent_channel}
\mathcal{K}({\Psi})=\{ X \in \mathcal{B}(\mathcal{H}) \,,\, \Psi( X^\ast X)= \Psi( X X^\ast )=0  \} = \Ker(\Psi) \cap \mathcal{N}({\Psi}),
\end{equation}
with $\mathcal{N}({\Psi})$ being the decoherence-free algebra of $\Psi$.  

Moreover, $\Nstar$ is a ${C}^\ast$-algebra with respect to the composition product $\cdot$ and the operator norm $\| \cdot \|$. Also, $\mathcal{K}({\mathcal{P}_{\mathrm{P}}})$ is a $\ast$-ideal of $\Nstar$ and, consequently, the quotient space
\begin{equation}
\label{Q_def}
\mathcal{Q} := \Nstar / \mathcal{K}({\mathcal{P}_{\mathrm{P}}})
\end{equation}
is a ${C}^\ast$-algebra $\ast$-automorphic to $(\Attr(\Phi),\star,\norm{\ \cdot\ })$. 
\end{Theorem}

\begin{proof}
The result was already proved by Hamana for idempotent UCP maps, see~\cite[Theorem 2.3]{hamana1979injective}.  
By applying this result to $\mathcal{P}_{\mathrm{P}}$, the peripheral projection of $\Phi$, we obtain that $\Attr(\Phi) \oplus \mathcal{K}({\mathcal{P}_{\mathrm{P}}}) = \mathcal{N}_{\star}(\mathcal{P}_{{\mathrm{P}}})$, the Choi-Effros decoherence-free algebra of $\mathcal{P}_{\mathrm{P}}$, is a ${C}^\ast$-algebra with respect to the composition product $\cdot$ and the operator norm $\| \cdot \|$. That being said, let us prove that $\Nstar=\Attr(\Phi)\oplus\mathcal{K}({\mathcal{P}_{\mathrm{P}}})$.

We start by proving $\Nstar \subseteq \Attr(\Phi) \oplus \mathcal{K}({\mathcal{P}_{\mathrm{P}}})$. To this purpose, we only need to show that, given $X \in\Nstar$, we have $X-\mathcal{P}_{\mathrm{P}}(X) \in \mathcal{K}({\mathcal{P}_{\mathrm{P}}})$. From Eq.~\eqref{P_p_for}, the condition
\begin{equation}
\Phi^n(X^\ast \star X)=\Phi^n(X)^\ast \star \Phi^n(X), \quad n\in\mathbb{N}, 
\end{equation}
implies that
\begin{equation}
	\label{N_in_M_P}
	\mathcal{P}_{\mathrm{P}}(X^\ast \star X)=X^\ast\star X=\mathcal{P}_{\mathrm{P}}(X)^\ast \star \mathcal{P}_{\mathrm{P}}(X),
\end{equation}
as well as the analogous equality with $X \leftrightarrow X^\ast$.
Therefore
\begin{equation}
\begin{split}
&\mathcal P_{\mathrm{P}}((X-\mathcal{P}_{\mathrm{P}}(X))^\ast  (X-\mathcal{P}_{\mathrm{P}}(X))) \\&= X^\ast \star X+\mathcal{P}_{\mathrm{P}}(X)^\ast\star \mathcal{P}_{\mathrm{P}}(X)- \mathcal{P}_{\mathrm{P}}(X^\ast \mathcal{P}_{\mathrm{P}}(X))-\mathcal{P}_{\mathrm{P}}(\mathcal{P}_{\mathrm{P}}(X)^\ast  X)\\
&= 2 \mathcal{P}_{\mathrm{P}}(X)^\ast\star \mathcal{P}_{\mathrm{P}}(X) - 2 \mathcal{P}_{\mathrm{P}}(X)^\ast\star \mathcal{P}_{\mathrm{P}}(X) =0,
\end{split}
\end{equation}
where in the second last step we used~\eqref{N_in_M_P} and Lemma~\ref{lemma_Hamana}. The proof of  
\begin{equation}
\mathcal{P}_{\mathrm{P}}((X-\mathcal{P}_{\mathrm{P}}(X))\star (X-\mathcal{P}_{\mathrm{P}}(X))^\ast) =0
\end{equation} 
is analogous.



Conversely, consider $X\in\Attr(\Phi)$ and $K\in\mathcal K({\mathcal P_{\mathrm{P}}})$, and let us prove that $X+K\in\Nstar$. We first prove that
\begin{equation}
	\label{aut_K_P}
	\Phi^n(K^\ast \star K)=\Phi^n(K)^\ast \star \Phi^n(K)=0,
\end{equation}
The left-hand side vanishes because $K^\ast\star K=0$ as $K\in\mathcal{K}({\mathcal{P}_{\mathrm{P}}})$, while the second term is $0$ by applying $n$ times the $\star$-operator Schwarz inequality~\eqref{star_op_Schw_ineq}:
\begin{equation}
\label{n_Schwarz_arg}
\begin{split}
	0&=\Phi^{n}(K^\ast \star K) \geqslant \Phi^{n-1}(\Phi(K)^\ast \star \Phi(K))\geqslant \dots\geqslant \Phi^n(K)^\ast \star \Phi^n(K) \\&\Rightarrow \Phi^n(K)^\ast \star \Phi^n(K)=0.
\end{split}
\end{equation}
Therefore, we can write
\begin{equation}
\begin{split}
	&\Phi^n( (X + K)^\ast \star (X + K))
	= \Phi^{n}(X^\ast \star X + X^\ast \star K + K^\ast \star X + K^\ast \star K)\\
	&=\Phi^{n}(X)^\ast \star \Phi^{n}(X) + \Phi^{n}(X)^\ast \star \Phi^{n}(K) +\Phi^{n}(K)^\ast \star \Phi^{n}(X)\\
	&= \Phi^n(X+K)^\ast \star \Phi^n(X + K), \quad n \in\mathbb{N}, 
\end{split}
\end{equation}
where we used Theorem~\ref{aut_th_UCP}, condition \textit{(b)} of Theorem~\ref{mul_ch_UCP}, and Eq.~\eqref{aut_K_P}.
Analogously, we can prove that
\begin{equation}
	\Phi^n( (X  + K) \star (X + K)^\ast) = \Phi^n(X+K)\star \Phi^n(X + K)^\ast.
\end{equation} 
Therefore, we get $\Nstar = \Attr(\Phi)\oplus \mathcal K({\mathcal P_{\mathrm{P}}})$.

Also, $\mathcal{K}({\mathcal{P}_{\mathrm{P}}})$ is a $\ast$-ideal of $\Nstar$ by \cite[Theorem 2.3]{hamana1979injective}.
Now, consider the quotient space $\mathcal{Q}$ defined by Eq.~\eqref{Q_def}, whose elements have the form
\begin{equation}
\label{equiv_cl}
\lfloor X \rfloor = \{  X + K\,,\, K\in \mathcal{K}({\mathcal{P}_{\mathrm{P}}}) \},
\end{equation}
with $X$ in $\Nstar$. Also, a product, an involution, and a norm are naturally defined for all  $\lfloor X \rfloor,\lfloor Y \rfloor$ in $\mathcal{Q}$ as
\begin{align}
\lfloor X\rfloor \cdot \lfloor Y \rfloor &:= \lfloor X Y \rfloor,\label{prod_Q_def}\\
\lfloor X \rfloor^\ast &:= \lfloor X^\ast \rfloor,\nonumber\\
\| \lfloor X \rfloor \| &:= \inf \{ \| X+K \|\,,\, K\in \mathcal{K}({\mathcal{P}_{\mathrm{P}}}) \}.
\label{norm_Q_def}
\end{align}
The map $\phi$
\begin{equation}
\label{iso_Attr}
\begin{split}
\phi:& \Attr(\Phi) \rightarrow \mathcal{Q} \\
& X \mapsto \phi(X)=\lfloor X \rfloor 
\end{split}
\end{equation}
is a $\ast$-automorphism between the $C^\ast$-algebras $(\Attr(\Phi) , \star , \| \cdot \|)$ and $(\mathcal{Q} , \cdot , \| \cdot \|)$. In particular, given $X,Y \in\Attr(\Phi)$, we have
\begin{equation}
\label{phi_auto}
\begin{split}
\phi(X \star Y)
&= \{ X \star Y + K\,,\, K\in \mathcal{K}({\mathcal{P}_{\mathrm{P}}}) \} = \{ XY - XY + X \star Y +K\,,\, K\in \mathcal{K}({\mathcal{P}_{\mathrm{P}}})\} \\
&=  \{ XY + K\,,\, K\in \mathcal{K}({\mathcal{P}_{\mathrm{P}}}) \} = \phi(X)\cdot \phi(Y).
\end{split}
\end{equation}
Here the third equality follows from the fact that $XY \in \Nstar$, so
\begin{equation}
XY - X \star Y = XY - \mathcal{P}_{\mathrm{P}}(XY) \in \mathcal{K}({\mathcal{P}_{\mathrm{P}}})
\end{equation}
by the decomposition~\eqref{sum_dec_N_star}. 
In particular, $\Attr (\Phi)$ and $\mathcal Q$ are isomorphic $C^\ast$-algebras, as $\phi$ is faithful. Indeed, if $\phi(X)=0$, then $\lfloor X \rfloor=\lfloor 0 \rfloor$, and $X\in\mathcal{K}({\mathcal{P}_{\mathrm{P}}})$. Thus, $X=0$ by~\eqref{sum_dec_N_star}.
\end{proof}
An immediate and important corollary of Theorem~\ref{sum_dec_N_th} is the following.
\begin{Corollary}
\label{N_star_in_N_star_P}
Let $\Phi$ be a UCP map with peripheral projection $\mathcal{P}_{\mathrm{P}}$, and Choi-Effros decoherence-free algebra $\Nstar$. Then
\begin{equation}
	\label{eq:N_star_equal_N}
\Nstar 
=  \mathcal{N}_{\star}(\mathcal{P}_{\mathrm{P}}),
\end{equation}
where $\mathcal N_{\star}(\mathcal{P}_{\mathrm{P}})$ is the Choi-Effros decoherence-free algebra of the UCP map $\mathcal P_{\mathrm{P}}$.
\end{Corollary}

The previous corollary highlights the connection between the algebra $\Nstar$ and the peripheral projection $\mathcal{P}_{\mathrm{P}}$ and, consequently, $\Attr(\Phi)$.
	
\begin{Remark}
Corollary~\ref{N_star_in_N_star_P} does not hold for $\mathcal{N}$ and $\mathcal{N}({\mathcal{P}_{\mathrm{P}}})$, the decoherence-free algebra of $\Phi$ and $\mathcal{P}_{\mathrm{P}}$, respectively. More precisely,
\begin{equation}
\label{N_in_N_P}
\mathcal{N} \subseteq \mathcal{N}({\mathcal{P}_{\mathrm{P}}}),
\end{equation}
where the inclusion may be strict, as it can be seen with the map~\eqref{Mark_pa_ex} in Example~\ref{ex_2}. We will see that the equality holds in the faithful case, see Theorem~\ref{faith_mul_ch}.
\end{Remark}

The previous construction shows that the attractor subspace is ultimately connected with the Choi-Effros decoherence-free algebra. We are now going to show that this connection can also be expressed in terms of a suitable GNS construction. To this purpose, we will 
construct the GNS Hilbert space of $\Nstar$ corresponding to a suitable state.

Specifically, define the functional $\omega : \Nstar \rightarrow \mathbb{C} $ as
\begin{equation}
\omega(X)=\frac{1}{d} \mathrm{tr}(\mathcal{P}_{\mathrm{P}}^\dagger(\mathbb{I})X)= \frac{1}{d} \tr(\mathcal{P}_{\mathrm{P}}(X)), 
\end{equation} 
being the state on $\Nstar$ associated with the density operator $ \frac{1}{d} \mathcal{P}_{\mathrm{P}}^\dagger (\mathbb{I})$. Recall that it is a maximum-rank asymptotic state of $\Phi^\dagger$, the Hilbert-Schmidt adjoint of $\Phi$. The corresponding GNS sesquilinear form,
\begin{equation}
\braket{X}{Y}_\omega= \omega(X^\ast Y) = \frac{1}{d} \mathrm{tr}(\mathcal{P}_{\mathrm{P}}^\dagger(\mathbb{I})X^\ast Y),
\end{equation} 
is not generally a scalar product since it may be degenerate. More precisely, we can introduce 
\begin{equation}
\mathcal{K} = \{ X \in \Nstar \,,\, \omega(X^\ast X) = 0 \},
\end{equation}
which is not trivial in general. Notice that $\mathcal{K}$ is equal to $\mathcal{K}({\mathcal{P}_{\mathrm{P}}})$ as defined in~\eqref{eq:definition_K_idempotent_channel}. Indeed, given $X \in \Nstar$, we have
\begin{equation}
\begin{split}
X \in \mathcal{K} &\Leftrightarrow \omega(X^\ast X) = \frac{1}{d} \mathrm{tr}(\mathcal{P}_{\mathrm{P}}^\dagger (\mathbb{I}) X^\ast X) = \frac{1}{d} \tr(\mathcal{P}_{\mathrm{P}}(X^\ast X))=0 \\&\Leftrightarrow \mathcal{P}_{\mathrm{P}}(X^\ast X) =0 \Leftrightarrow X \in \mathcal{K}({\mathcal{P}_{\mathrm{P}}}), 
\end{split}
\end{equation}
where in the second last equivalence we used the positivity of $\mathcal{P}_{\mathrm{P}}$. Therefore, as we did in the proof of Theorem~\ref{sum_dec_N_th},  consider the quotient space
\begin{equation}
\mathcal{Q} = \Nstar / \mathcal{K}({\mathcal{P}_{\mathrm{P}}}) = \Nstar / \mathcal{K} \simeq \Attr(\Phi).
\end{equation}
The scalar product will explicitly become
\begin{equation}
\braket{ \lfloor X \rfloor }{ \lfloor Y \rfloor}_{\omega} = \omega({X}^\ast {Y}).  
\end{equation}
Here, $\lfloor X \rfloor$ is the equivalence class in $\mathcal Q$ containing $X$. Therefore the quotient space $\mathcal{Q}$ equipped with $\braket{\cdot }{ \cdot }_{\omega}$ is the GNS Hilbert space $\mathcal{H}_{\omega}$ corresponding to the state $\omega$ and the algebra $\Nstar$. 
The GNS representation $\pi_\omega$ of $\Nstar$ 
reads
\begin{equation}
\pi_{\omega}(X)\lfloor Y \rfloor = \lfloor XY \rfloor = \lfloor X \rfloor \cdot \lfloor Y \rfloor, \quad X,Y\in \Nstar.
\end{equation}
Therefore we find again the natural product $\cdot$ on $\mathcal{Q}$ induced by the composition product endowing $\Nstar$. In particular, recall that, according to~\eqref{phi_auto}, this corresponds to the Choi-Effros product in $\Attr(\Phi)$ via the isomorphism $\phi$ defined by Eq.~\eqref{iso_Attr}. 

Notice that, if $\Phi$ is faithful, then $\mathcal{P}_{\mathrm{P}}^\dagger(\mathbb{I})>0$, $\mathcal{K} = \{ 0 \}$, and $\Attr(\Phi)=\Nstar$, and, in fact, as we will show in Theorem~\ref{faith_mul_ch}, this condition characterizes the faithfulness of $\Phi$.

A third way to recover, up to a $\ast$-automorphism, the algebraic structure of $\Attr(\Phi)$ from the algebra $\Nstar$ will be discussed in Appendix~\ref{Nstar_seminorm}. In particular, the product~\eqref{prod_Q_def} turns out to be equal to an analogous one~\eqref{Choi_quot_prod} induced by the Choi-Effros product on $\Nstar$. 
 
Now, in the discussion after Theorem~\ref{C_star_UCP}, we saw that the Choi-Effros product does not behave well on $\mathcal B(\mathcal H)$, as it is not associative in general. However, we are now going to see that this is true when restricted to $\Nstar$, making it a $\ast$-algebra with respect to the Choi-Effros product.
\begin{Lemma}
\label{asso_N_CE_prop}
Let $\Phi$ be a UCP map. Then the Choi-Effros product defined by Eq.~\eqref{star_prod_H} is associative within its Choi-Effros decoherence-free algebra $\Nstar$, i.e.
\begin{equation}
(X \star Y) \star Z = X \star (Y \star Z), \quad X,Y,Z\in \Nstar.
\end{equation}
\end{Lemma}
\begin{proof}

First, given $X,Y \in \Nstar$, a direct consequence of Corollary~\ref{N_star_in_N_star_P} and Lemma~\ref{lemma_Hamana} is the following:
\begin{equation}
\mathcal{P}_{\mathrm{P}}(X Y)=\mathcal{P}_{\mathrm{P}}(\mathcal{P}_{\mathrm{P}}(X) \mathcal{P}_{\mathrm{P}}(Y))= \mathcal{P}_{\mathrm{P}}(\mathcal{P}_{\mathrm{P}}(X) Y)=\mathcal{P}_{\mathrm{P}}(X\mathcal{P}_{\mathrm{P}}(Y)).
\end{equation}
Thus, it follows that
\begin{equation}
\begin{split}
X \star (Y \star Z)&= \mathcal{P}_{\mathrm{P}}(X\mathcal{P}_{\mathrm{P}}(YZ))=\mathcal{P}_{\mathrm{P}}(\mathcal{P}_{\mathrm{P}}(X)YZ)=\mathcal{P}_{\mathrm{P}}(XYZ)\\&= \mathcal{P}_{\mathrm{P}}(XY\mathcal{P}_{\mathrm{P}}(Z))= \mathcal{P}_{\mathrm{P}}(\mathcal{P}_{\mathrm{P}}(XY)Z)=(X \star Y) \star Z.
\end{split}
\end{equation}
for any $X,Y,Z \in \mathcal{N}_{\star}$.
Here, we employed the fact that the composition product makes $\Nstar$ a $C^\ast$-algebra.
\end{proof}
In the next theorem, we are going to see that $\Nstar$ carries a $B^\ast$-algebraic structure with respect to the Choi-Effros product.

\begin{Theorem}
\label{C*-dec_th}
Let $\Phi$ be a UCP map. Then its Choi-Effros decoherence-free algebra $\Nstar$ is a Banach $\ast$-algebra with respect to the Choi-Effros product $\star$ defined by Eq.~\eqref{star_prod_H} and the operator norm $\| \cdot \|$, or $(\Nstar , \star , \| \cdot \|)$ in short. 
\end{Theorem}
\begin{proof}

Condition \textit{b)} of Theorem~\ref{mul_ch_UCP} and Lemma~\ref{asso_N_CE_prop} imply that $\Nstar$ is a $\ast$-algebra with respect to the Choi-Effros product $\star$. Also, $\Nstar$ is a Banach space with respect to the operator norm $\| \cdot \|$. Therefore it only remains to prove the submultiplicativity of $\| \cdot \|$ with respect to the Choi-Effros product, which follows from the contractivity of $\mathcal{P}_{\mathrm{P}}$
\begin{equation}
\| X \star Y \| = \| \mathcal{P}_{\mathrm{P}}(XY) \| \leqslant \| X \| \| Y \|, \quad X,Y \in \Nstar.
\end{equation}  
\end{proof}

\begin{Remark}
\label{Nstar_Cast_rem}
Importantly, $(\Nstar , \star , \| \cdot \|)$ is not generally a $C^\ast$-algebra, as it does not satisfy the $C^\ast$-identity, see Example~\ref{ex_1}.
\end{Remark}


Once the algebraic structure of $\Nstar$ has been understood, we can delve into its connection with $\Attr(\Phi)$ and $\mathcal{N}$.
In this respect, the next theorem shows that the space $\Nstar$ reduces to $\mathcal{N}$ in the faithful case and, importantly, $\Nstar$ allows to characterize faithful UCP maps in terms of their multiplicative properties.
\begin{Theorem}[Multiplicative characterization of faithfulness]
\label{faith_mul_ch}
Let $\Phi$ be a UCP map with attractor subspace $\Attr(\Phi)$ and Choi-Effros decoherence-free algebra $\Nstar$. Then
\begin{equation}
\label{mul_ch_faith}
\Phi \mbox{ faithful} \quad \Leftrightarrow \quad \Attr(\Phi)=\Nstar. 
\end{equation}
Furthermore, in such case, $\Nstar=\mathcal{N} = \mathcal{N}(\mathcal{P}_{\mathrm{P}})$.
\end{Theorem} 
\begin{proof}
Let us first prove that $\Nstar = \mathcal{N}$ in the faithful case.
From property~\textit{a)} of Proposition~\ref{ch_N_star_th}, we only need to prove $\Nstar\subseteq \mathcal N$. Given $X\in \Nstar$ and $n \in \mathbb{N}$, we can write
\begin{equation}
\Phi^n(X^\ast\star X)=\Phi^n(X)^\ast \star \Phi^n(X).
\end{equation}
Therefore, given $\rho$ invertible invariant state of $\Phi^\dagger$, namely $\mathcal{P}_{\mathrm{P}}^\dagger(\rho)=\Phi^\dagger(\rho)=\rho$, we have
\begin{equation}
\begin{split}
0&=\tr(\rho(\Phi^n(X^\ast\star X)-\Phi^n(X)^\ast \star \Phi^n(X)))= \mathrm{tr}(\rho\mathcal{P}_{\mathrm{P}}(\Phi^n(X^\ast X)-\Phi^n(X)^\ast\Phi^n(X)))\\
&= \mathrm{tr}(\mathcal P_{\mathrm{P}}^\dagger(\rho)(\Phi^n(X^\ast X)-\Phi^n(X)^\ast\Phi^n(X)))=\tr(\rho(\Phi^n(X^\ast X)-\Phi^n(X)^\ast\Phi^n(X))).
\end{split}
\end{equation}
Since $\Phi^n(X^\ast X)\geqslant \Phi^n(X^\ast)\Phi^n(X)$ and $\rho$ is invertible, this implies
\begin{equation}
\Phi^n(X^\ast X)=\Phi^n(X)^\ast\Phi^n(X).
\end{equation}
The other equality follows by exchanging $X$ with $X^\ast$. Therefore, the implication $\Rightarrow)$ follows by combining the just proved equality and Theorem~\ref{attr_ch_faith}.

Let us now prove the $\Leftarrow )$ implication. 
%
%
From the direct-sum decomposition~\eqref{sum_dec_N_star} the hypothesis is equivalent to 
\begin{equation}
	\label{eq:faithful_condition}
\mathcal{K}({\mathcal{P}_{\mathrm{P}}})=\{ 0 \}.
\end{equation}

Remember that a channel is faithful if and only if $\supp(\mathcal P_{\mathrm{P}}^\dagger(\mathbb I))$, the support of $\mathcal P_{\mathrm{P}}^\dagger(\mathbb I)$, is equal to $\mathcal H$ (see~\cite[Proposition 6.9]{wolf2012quantum} and~\cite[Appendix A.2]{AFK_asympt_2}). Suppose that $\supp(\mathcal P_{\mathrm{P}}^\dagger(\mathbb I))\subset \mathcal H$, and take $X=\ketbra{\varphi}$, with $\ket{\varphi}$ a normalized vector such that $\mathcal P_{\mathrm{P}}^\dagger(\mathbb I) \ket{\varphi}=0$. Then 
\begin{equation}
\tr(\mathcal{P}_{\mathrm{P}}(X^\ast X))=\mathrm{tr}(\mathcal{P}_{\mathrm{P}}^\dagger (\mathbb{I}) X^\ast X) = \mel{\varphi}{\mathcal P_{\mathrm{P}}^\dagger(\mathbb I)}{\varphi}=0.
\end{equation}
As a result, $\mathcal P_{\mathrm{P}}(X^\ast X)=\mathcal P_{\mathrm{P}}(XX^\ast)=\mathcal P_{\mathrm{P}}(X)=0$, i.e. $X\in \mathcal{K}({\mathcal{P}_{\mathrm{P}}})$, which contradicts~\eqref{eq:faithful_condition}.
 This necessarily implies that $\mathcal{P}_{\mathrm{P}}^\dagger(\mathbb{I})>0$, meaning that $\Phi$ is faithful. Finally, the equality $\mathcal{N} = \mathcal{N}(\mathcal{P}_{\mathrm{P}})$ is a simple consequence of Eq.~\eqref{eq_Attr_N_faith} and the definition~\eqref{P_p_def} of $\mathcal{P}_{\mathrm{P}}$.
\end{proof}
\begin{Remark}
We will see in Example~\ref{ex_2} that $\Attr(\Phi)=\mathcal{N}$ does not imply the faithfulness of $\Phi$.
\end{Remark}
\begin{Remark}
\label{Phi_N_not_aut}
Clearly, $\Phi\vert_{\Nstar}$ is a $\ast$-homomorphism of the $B^\ast$-algebra $(\Nstar , \star , \| \cdot \|)$. However, as we will see later in Example~\ref{ex_1}, it is not generally an automorphism. As it happens for $\mathcal{N}$, $\Phi \vert_{\Nstar}$ is always invertible in the faithful case as a consequence of Theorem~\ref{faith_mul_ch}. However, it should be noted that there are non-faithful maps, including invertible ones, that also possess this property, as illustrated in Example~\ref{ex_2}.
\end{Remark}
From the previous theorem, the following corollary holds.
\begin{Corollary}
\label{faith_cor}
Let $\Phi$ be a UCP map. If $\Attr(\Phi)=\Nstar$, then we have
\begin{equation}
\Attr(\Phi)=\Nstar =\mathcal{N}.
\end{equation}
\end{Corollary}
\begin{proof}
Any UCP map satisfying the assumption is faithful by~\eqref{mul_ch_faith}. Thus, the assertion follows from the second part of Theorem~\ref{faith_mul_ch}.
\end{proof}
Let us summarize the inclusions between $\Attr(\Phi)$, $\mathcal{N}$ and $\Nstar$ in the following theorem.
\begin{Theorem}
\label{sum_up_incl}
Let $\Phi$ be a UCP map. Let $\mathcal{N}$ be its decoherence-free algebra and let $\Nstar$ be its Choi-Effros decoherence-free algebra. Then
\begin{align}
&\Phi \mbox{ faithful} &\Leftrightarrow\quad &\Attr(\Phi)=\mathcal{N}=\Nstar,\nonumber\\
&\Phi \mbox{ peripherally automorphic} &\Leftrightarrow\quad &\Attr(\Phi)\subseteq \mathcal{N} \subseteq \Nstar,\nonumber\\
&\Phi \mbox{ not peripherally automorphic} &\Rightarrow\quad &\Attr(\Phi)  \subset \Nstar, \Attr(\Phi) \not\subseteq \mathcal{N}, \; \mathcal N\subseteq \Nstar.
\end{align}

\end{Theorem}
We conclude the present section with the following lemma, which shows that the equality $\Nstar=\mathcal{N}$ is a necessary but not sufficient condition for the faithfulness of $\Phi$.
\begin{Lemma}
\label{idem_con}
If $\Phi$ is a peripherally automorphic idempotent UCP map, then it follows that
\begin{equation}
\Nstar = \mathcal{N}.
\end{equation}
\end{Lemma}
\begin{proof}
Taking $X\in\Nstar$, and using condition $b)$ of Theorem~\ref{ch_per_aut_1}, we have
\begin{equation}
\Phi(X^\ast X)= \Phi(\Phi(X)^\ast\Phi(X))=\Phi(X)^\ast\Phi(X).
\end{equation} 
Since $\Phi^n = \Phi$, we have $X\in \mathcal{N}$, i.e. the assertion.
\end{proof}

\subsection{The qubit case}
\label{sec:qubit}

We now discuss a couple of results on the asymptotics of qubit UCP maps (i.e.\ $d=\dim\mathcal{H}=2$), highlighting the peculiarities of the two-dimensional case.

\begin{Proposition}
Let $\Phi$ be a non-faithful qubit UCP map. Then its attractor subspace $\Attr(\Phi)$ is one-dimensional
\begin{equation}
\Attr(\Phi)=\{ c\mathbb{I} \,,\, c\in\mathbb{C} \}, 
\end{equation}
and its peripheral projection $\mathcal{P}_{\mathrm{P}}$ reads
\begin{equation}
\label{one_proj}
\mathcal{P}_{\mathrm{P}}(X)=\tr(\rho X)\mathbb{I},
\end{equation}
where $\Phi^\dagger(\rho)=\rho$ is the (non-invertible) invariant state of the adjoint map $\Phi^\dagger$.
\end{Proposition} 
\begin{proof}
Consider a qubit evolving under a UCP map $\Phi$, and call $m_{\mathrm{P}}=\dim(\Attr(\Phi))$. In this case, as a consequence of~\cite[Theorem 2]{AF_bounds}, we have
\begin{equation}
m_{\mathrm{P}}=1,2,
\end{equation}
where the latter value is admissible only for faithful UCP maps. The second part of the assertion follows by taking the adjoint of the one-state contraction channel
\begin{equation}
\mathcal{P}_{\mathrm{P}}^\dagger(X)=\tr(X)\rho,
\end{equation}
which is the only admissible one-dimensional idempotent quantum channel.
\end{proof}
An immediate consequence of the latter result is that the attractor subspace of any qubit UCP map is a $C^\ast$-algebra with respect to the composition product.
\begin{Corollary} 
\label{qubit_pa}
Let $\Phi$ be a qubit UCP map. Then $\Phi$ is peripherally automorphic.
\end{Corollary}

\section{Examples}
\label{examples_sec}
We will now consider some interesting examples in the qubit and qutrit cases in order to illustrate and clarify the previous results and remarks. 

First, observe that, by combining Theorem~\ref{sum_up_incl} and Corollary~\ref{qubit_pa}, we have $\Attr(\Phi) \subseteq \mathcal{N} \subseteq \mathcal{N}_{\star}$ for qubit UCP maps. Using the following two-dimensional examples, we will now discuss the strictness of these inclusions.
\begin{Example}
\label{ex_1_qubit}
Take the idempotent UCP map $\mathcal{P}$ on the space $\mathcal{M}_2$ of complex matrices of order 2
\begin{equation}
\label{1_qubit}
\mathcal{P}(X) = 
\begin{pmatrix}
x_{11} & \\
& x_{11}
\end{pmatrix} 
= \tr(\rho X) \mathbb{I}, \quad \rho =  \begin{pmatrix}
1 & \\
& 0
\end{pmatrix},
\end{equation}
where henceforth the unspecified elements of a matrix are zeros. Again, $\mathcal{P}$ is peripherally automorphic by Corollary~\ref{qubit_pa}, while it is not faithful since the only stationary state of $\mathcal{P}^{\dagger}$ is $\rho$. Moreover, 
\begin{equation}
\Attr(\mathcal{P}) = \Fix(\mathcal{P}) = \mathbb{C} \mathbb{I},
\end{equation}
whereas by direct computation the decoherence-free algebra $\mathcal{N}$ reads
\begin{equation}
\mathcal{N} = \left\{  X =  
\begin{pmatrix}
x_{11} &   \\
 & x_{22}  
\end{pmatrix} \,,\, x_{ii}\in\mathbb{C},\, i=1,2 \right\}.
\end{equation}
Finally, $\mathcal{N} = \mathcal{N}_{\star}$ by Lemma~\ref{idem_con}, and thus,
\begin{equation}
\label{sc_1}
\Attr(\mathcal{P}) \subset \mathcal{N} = \mathcal{N_{\star}}.
\end{equation} 
\end{Example}
\begin{Example}
\label{ex_2_qubit}
Consider the Markovian UCP map on $\mathcal{M}_2$ defined as
\begin{equation}
\Phi = e^{\mathcal{P} - \mathsf{1}},
\end{equation} 
with $\mathcal{P}$ given by~\eqref{1_qubit}.
Again, $\Phi$ is peripherally automorphic and, indeed, 
\begin{equation}
\Attr(\Phi) = \Attr(\mathcal{P}) = \mathbb{C} \mathbb{I}.
\end{equation}
In addition, it is easy to check that $\mathcal{N}=\Attr(\Phi)$, while 
\begin{equation}
\mathcal{N}_{\star} = \left\{  X =  
\begin{pmatrix}
x_{11} &   \\
 & x_{22}  
\end{pmatrix} \,,\, x_{ii}\in\mathbb{C},\, i=1,2 \right\}.
\end{equation}
Thus, 
\begin{equation}
\label{sc_2}
\Attr(\mathcal{P}) = \mathcal{N} \subset \mathcal{N_{\star}}.
\end{equation}
\end{Example}
As a consequence of Corollary~\ref{faith_cor}, in the faithful case $\Attr(\Phi) = \mathcal{N} = \mathcal{N}_{\star}$. We can show that the latter scenario together with the ones in Examples~\ref{ex_1_qubit} and~\ref{ex_2_qubit} [see Eqs.~\eqref{sc_1} and~\eqref{sc_2}] exhaust all the possibilities in the qubit case. To this purpose, we have to show that for a qubit UCP map we cannot have $\Attr(\Phi) \subset \mathcal{N} \subset \mathcal{N}_{\star}$, viz. the two inclusions are strict. 

Let $\Phi$ be a non-faithful UCP map. According to Theorem~\ref{sum_dec_N_th}, $\mathcal{N}_{\star}$ is a $C^\ast$-subalgebra of $\mathcal{B}(\mathcal{H})$, and its only possible dimensions are $1,2,4$ when $d=2$ by~\cite[p. 74]{Davidson1996}. However, the inclusions can be strict only when $\dim (\mathcal{N}_{\star}) = 4$, viz.\ $\mathcal{N}_{\star} = \mathcal{B}(\mathcal{H})$. Let us now show that, for any $d \geqslant 2$, the latter equality necessarily implies that $\Phi$ is unitary, in contrast with the non-faithfulness of $\Phi$.  
In fact, we have that
\begin{equation}
X \star Y = \mathcal{P}_{\mathrm{P}}(XY) = \mathcal{P}_{\mathrm{P}}(X) \star \mathcal{P}_{\mathrm{P}}(Y) 
, \quad \forall X,Y \in \mathcal{B}(\mathcal{H}),
\end{equation}
where we combined Lemma~\ref{lemma_lemme} 
and Corollary~\ref{N_star_in_N_star_P} in the second step. Thus, $\mathcal{P}_{\mathrm{P}}$ is a homomorphism between the C$^\ast$-algebras $(\mathcal{B}(\mathcal{H}) , \cdot , \| \cdot \|)$ and $(\Ran(\mathcal{P}_{\mathrm{P}}) , \star , \| \cdot \|)$, implying that $\Ker(\mathcal{P}_{\mathrm{P}}) \subseteq \mathcal{B}(\mathcal{H})$ is an ideal. Since $\mathcal{B}(\mathcal{H})$ has no proper ideals and $\mathcal{P}_{\mathrm{P}}(\mathbb{I}) = \mathbb{I}$, we conclude that  $\mathcal{P}_{\mathrm{P}}$ is invertible. Thus, this means that 
$\mathcal{P}_{\mathrm{P}} = \mathsf{1}$ 
and, consequently, $\Phi$ is unitary by~\cite[Theorem 6.1]{wolf2012quantum}. 

Let us now turn our attention to the qutrit case.
The first qutrit example is inspired by~\cite[Example 5.3]{fidaleo2022spectral}, and illustrates what happens for non-peripherally automorphic UCP maps.
\begin{Example}
\label{ex_1}
Consider an idempotent UCP map $\mathcal{P}$ on $\mathcal{M}_3$, the space of $3\times 3$ complex matrices 
\begin{equation}
\mathcal{P}(X)=\begin{pmatrix}
x_{11} & x_{12} & \\
x_{21} & x_{22} & \\
 &  & x_{11}
\end{pmatrix}.
\end{equation}
Clearly,
\begin{equation}
\label{ex_1_Attr}
\Attr(\mathcal{P})=\left\{ X = \begin{pmatrix}
x_{11} & x_{12} & \\
x_{21} & x_{22} & \\
 &  & x_{11}
\end{pmatrix} \,, \,  x_{ij}\in\mathbb{C}, i,j=1,2  \right\},
\end{equation}
from which it follows that $\mathcal{P}$ is not peripherally automorphic, as
\begin{equation}
\mathcal{P}(X^\ast X) \neq X^\ast X,
\end{equation}
whenever $X \in \Attr(\mathcal{P})$ with $x_{21}\neq 0$.

A straightforward computation shows that its decoherence-free algebra reads
\begin{align}
\mathcal{N} &= \left\{ X =  
\begin{pmatrix}
x_{11} &  & \\
 & x_{22} & \\
 &  & x_{33}
\end{pmatrix} \,,\, x_{ii}\in\mathbb{C},\, i=1,2,3
\right\}, 
\end{align} 
while its Choi-Effros decoherence-free algebra is
\begin{align}
\Nstar &= \left\{ X = 
\begin{pmatrix}
x_{11} & x_{12} & \\
x_{21} & x_{22} & \\
 &  & x_{33}
\end{pmatrix} \,,\,   x_{ij}\in\mathbb{C},\, i,j=1,2,\, x_{33} \in\mathbb{C} \right\}.
\label{ex_1_Nstar}
\end{align} 
Thus, $\mathcal{N}\subset \Nstar$ and $\Attr(\mathcal{P}) \subset \Nstar$, implying that the inclusions stated in property~\textit{a)} of Proposition~\ref{ch_N_star_th} and condition \textit{b)} of Theorem~\ref{mul_ch_UCP} can be strict.

According to Remark~\ref{Nstar_Cast_rem}, $\Nstar$ is not a ${C}^{\ast}$-algebra with respect to the Choi-Effros product and the operator norm $\| \cdot \|$, since
\begin{equation}
\| \mathcal{P}(X^\ast X) \| < \| X \|^2,
\end{equation}
whenever $X\in\mathcal{N}\subset \Nstar$ and $\abs{x_{33}}>\abs{x_{11}},\abs{x_{22}}$.

Also, the subspace $\mathcal{K}({\mathcal{P}_{\mathrm{P}}})$ in~\eqref{sum_dec_N_star} coincides with $\mathcal{K}({\mathcal{P}})$ and takes the form  
\begin{equation}
\label{ex_1_KPp}
\mathcal{K}({\mathcal{P}_{\mathrm{P}}})= \left\{ 
 \begin{pmatrix}
& & & \\
& & &\\
& & & x_{33}
\end{pmatrix}  \,, \, x_{33}\in\mathbb{C}
\right\},
\end{equation}
and we find again~\eqref{ex_1_Nstar}, using~\eqref{ex_1_Attr} and decomposition~\eqref{sum_dec_N_star}. 

Since 
\begin{equation}
\mathcal{K}({\mathcal{P}_{\mathrm{P}}}) \subseteq \Ker(\mathcal{P})\cap \Nstar,
\end{equation}
$\mathcal{P}\vert_{\Nstar}$ is a non-injective $\ast$-homomorphism with respect to the Choi-Effros product~\eqref{star_prod_H}, so it is not a $\ast$-automorphism, in line with Remark~\ref{Phi_N_not_aut}. 

\end{Example}
The last example, treating the peripherally automorphic case when $d=3$, is quite similar to Examples~\ref{ex_1_qubit} and~\ref{ex_2_qubit}.
\begin{Example}
\label{ex_2}
Let us consider the Markovian UCP map on $\mathcal{M}_3$
\begin{equation}
\label{Mark_pa_ex}
\Phi=e^{\mathcal{P}-\mathsf{1}},
\end{equation}
where the idempotent UCP map $\mathcal{P}$, already considered in~\cite[Example 2.11]{rajaramaperipheral_2}, reads
\begin{equation}
\label{P_def}
\mathcal{P}(X)=
\begin{pmatrix}
x_{11} & &\\
& x_{22} &\\
& & x_{11}
\end{pmatrix}.
\end{equation}
Immediately, one has
\begin{equation}
\Attr(\Phi)=\Attr(\mathcal{P})= \left\{ 
\begin{pmatrix}
x_{11} & &\\
& x_{22} &\\
& & x_{11}
\end{pmatrix} \,, \, x_{11},x_{22}\in\mathbb{C}
\right\},
\end{equation}
therefore $\Phi$ as well as $\mathcal{P}$ is peripherally automorphic because of condition \textit{b)} of Theorem~\ref{ch_per_aut_1}, although it is not faithful~\cite{rajaramaperipheral_2}. Note also that $\mathcal{P}$ is the peripheral projection of $\Phi$.

Moreover, we obtain
\begin{equation}
\Attr(\Phi)=\mathcal{N},
\end{equation}
which proves that
\begin{equation}
\label{attr_N_no_faith}
\Attr(\Phi)=\mathcal{N}\not\Rightarrow \Phi \mbox{ faithful}.
\end{equation}

Also, it turns out that
\begin{equation}
\label{ex_2-Nstar}
\Nstar=\left\{ X = \begin{pmatrix}
x_{11} & & \\
& x_{22} & \\
& & x_{33}
\end{pmatrix} \,,\, x_{ii}\in\mathbb{C}, i=1,2,3 
\right\},
\end{equation}
so, as we expect from Theorem~\ref{faith_mul_ch}, $\mathcal{N}\subset \Nstar$. Moreover, the ideal $\mathcal{K}({\mathcal{P}_{\mathrm{P}}})$ has exactly the form~\eqref{ex_1_KPp} and consistently we recover again~\eqref{ex_2-Nstar} with the aid of the direct-sum decomposition \eqref{sum_dec_N_star}. 

Finally, if we take the UCP map $\mathcal{P}$ defined by Eq.~\eqref{P_def}, then its decoherence-free algebra $\mathcal{N}({\mathcal{P}})$ has the structure~\eqref{ex_2-Nstar} by Lemma~\ref{idem_con}, so
\begin{equation}
\mathcal{N}=\Attr(\mathcal{P}) \subset \mathcal{N}({\mathcal{P}}).
\end{equation}  
Therefore the inclusions \textit{b)} of Theorem~\ref{mul_ch_pa_UCP} and Eq.~\eqref{N_in_N_P} can be strict.
\end{Example}

\section{Is complete positivity needed?}
\label{Schw_CP}
The answer is no: in this section we will generalize Theorem~\ref{C_star_UCP} and Theorem~\ref{aut_th_UCP} to Schwarz maps, showing that all the previous results are valid for this larger class of maps. Therefore, complete positivity and, more generally, $k$-positivity with $k\geqslant 2$ is not used in their full power for characterizing the asymptotic dynamics of an open quantum system, see~\cite[Section 1]{Frigerio_Verri_82} and~\cite[Section V]{wolf2010inverse} for an analogous comment in the Schr\"odinger picture of the evolution.

A positive unital map $\Phi$ is called strongly positive~\cite{robinson1982strongly}, 1-1/2 positive~\cite{kielanowicz2017spectral}, or \textit{Schwarz}~\cite{wolf2010inverse}, if and only if it satisfies the operator Schwarz inequality~\eqref{op_Schw_ineq}.
(Note that in the literature the unitality assumption is sometimes dropped in the definition of Schwarz maps, see e.g.~\cite{siudzinska2021interpolating}). Similarly, we will call a positive unital map obeying the $\star$-operator Schwarz inequality~\eqref{star_op_Schw_ineq} a \emph{$\star$-Schwarz map}. 

The transposition map is an example of a positive unital map not obeying Eq.~\eqref{op_Schw_ineq}. As underlined by Choi~\cite[Corollary 2.8]{choi1974schwarz}, any 2-positive unital map is a Schwarz map, and, more generally, a 2-positive map satisfies a generalized operator Schwarz inequality, see~\cite[Corollary 7]{carlen2022characterizing}. However, 2-positivity is not necessary for~\eqref{op_Schw_ineq} to hold, see~\cite{choi1980some,chruscinski2020kadison}. Furthermore, a concept lying between the operator Schwarz inequality and 2-positivity was recently introduced~\cite{carlen2022characterizing}. See~\cite{AFM_2025} for a hierarchy of intermediate notions between positivity and complete positivity.
\subsection{Choi-Effros product for Schwarz maps}
\label{C_E_Schwarz}
According to~\cite[Theorem 3.1]{choi1977injectivity}, the attractor subspace $\Attr(\Phi)$ of any unital 2-positive map is a ${C}^\ast$-algebra with respect to the Choi-Effros product~\eqref{star_prod_H} and the operator norm $\| \cdot \|$, where the projection $\mathcal{P}_{\mathrm{P}}$ is again defined by~\eqref{P_p_def} and admits the limit representation~\eqref{P_p_for}, see~\cite[Proposition 6.3]{wolf2012quantum}. Indeed, the result holds for any Schwarz map $\Phi$ (see~\cite[Theorem 2.3]{hamana1979injective} or~\cite[Corollary  2.3.69]{garcia2014non}). See also~\cite[Theorem 7]{AFK_asympt_4} for the explicit structure of the attractor subspace of Schwarz maps.
\begin{Proposition}
\label{Schw_alg_attr}
Let $\Phi$ be a Schwarz map. Then its attractor subspace $\Attr(\Phi)$ defined by~\eqref{attr_def} is a ${C}^\ast$-algebra with respect to the Choi-Effros product $\star$~\eqref{star_prod_H} and the operator norm $\| \cdot \|$, or in short $(\Attr(\Phi) , \star , \| \cdot \|)$ is a $C^\ast$-algebra.
\end{Proposition}
\begin{Remark}
Note that the operator Schwarz inequality is only needed to prove the associativity of the Choi-Effros product and the ${C}^\ast$-identity on $\Attr(\Phi)$. Indeed, even the weaker $\star$-operator Schwarz inequality~\eqref{star_op_Schw_ineq} is sufficient in order to prove associativity. 
\end{Remark}
\begin{Remark}
As already observed, a Schwarz map is not generally 2-positive. Indeed, even if we focus our attention on unital positive projections, 2-positivity is equivalent to the Schwarz property only when $d=2,3$~\cite[Example pg 77, Theorem 3.1]{osaka1991positive}.

\end{Remark}
\subsection{Asymptotics of Schwarz maps}
According to~\cite[Theorem 2.12]{rajarama2022peripheral}, the asymptotic map $\Phi_{\mathrm{as}}$ of a UCP map $\Phi$ defined by Eq.~\eqref{as_map_def} is a $\ast$-automorphism. Indeed, the result may be generalized to Schwarz maps.
\begin{Theorem}
Let $\Phi$ be a Schwarz map. Then its asymptotic map $\Phi_{\mathrm{as}}: \Attr(\Phi) \mapsto \Attr(\Phi)$ defined by Eq.~\eqref{as_map_def} is a $\ast$-automorphism, i.e. $\Phi_{\mathrm{as}}$ is invertible and
\begin{equation}
\Phi_{\mathrm{as}}(X \star Y)=\Phi_{\mathrm{as}}(X) \star \Phi_{\mathrm{as}}(Y), \quad X,Y\in \Attr(\Phi).
\end{equation}
\end{Theorem}
\begin{proof}
First, notice that the asymptotic map $\Phi_{\mathrm{as}}$ is a $\star$-Schwarz map, since $\Phi$ is. Given $\Phi_{\mathrm{as}}(X)=\Phi(X)=\lambda X$, with $\lambda \in \spec_{\mathrm{P}}(\Phi)$, by using the $\star$-operator Schwarz inequality~\eqref{star_op_Schw_ineq} we have 
\begin{equation}
\Phi_{\mathrm{as}}(X^\ast \star X) \geqslant \Phi_{\mathrm{as}}(X)^\ast \star \Phi_{\mathrm{as}} (X) = X^\ast \star X.
\end{equation}
By~\eqref{P_p_for} the inverse map $\Phi_{\mathrm{as}}^{-1}: \Attr(\Phi) \mapsto \Attr(\Phi)$ of $\Phi_{\mathrm{as}}$ may be obtained as a limit
\begin{equation}
\label{as_for_inverse}
\Phi_{\mathrm{as}}^{-1}= \lim_{i\rightarrow \infty} \Phi_{\mathrm{as}}^{n_i-1},
\end{equation}
which makes it a $\star$-Schwarz map too. As a result
\begin{equation}
\Phi_{\mathrm{as}}^{-1}(X^\ast \star X) \geqslant \Phi_{\mathrm{as}}^{-1}(X)^\ast \star \Phi_{\mathrm{as}}^{-1} (X) = X^\ast \star X \Rightarrow X^\ast \star X \geqslant \Phi_{\mathrm{as}}(X^\ast \star X),
\end{equation}
and thus
\begin{equation}
\Phi_{\mathrm{as}} (X^\ast \star X)= X^\ast \star X =\Phi_{\mathrm{as}}(X)^\ast\star \Phi_{\mathrm{as}}(X),
\end{equation}
i.e.\ $X$ saturates the $\star$-operator Schwarz inequality~\eqref{star_op_Schw_ineq}, implying the assertion by Lemma~\ref{lemma_lemme}. 
\end{proof}

An immediate consequence of the latter theorem is the following corollary.
\begin{Corollary}
Let $\Phi$ be a Schwarz map with asymptotic map $\Phi_{\mathrm{as}}$. Then
\begin{equation}
\Phi_{\mathrm{as}} \mbox{ is a $\star$-Schwarz map } \Leftrightarrow\; \Phi_{\mathrm{as}} \mbox{ is a UCP map} .
\end{equation}
Here, complete positivity is meant with respect to the $C^\ast$-algebra $(\Attr(\Phi) , \star , \| \cdot \|)$.
\end{Corollary}
The latter corollary states precisely that the $\star$-operator Schwarz inequality and the complete positivity (with respect to the algebraic structure of $\Attr(\Phi)$) are equivalent at the level of the asymptotic map!
\begin{Remark}
Theorem~\ref{C_star_UCP} and Theorem~\ref{aut_th_UCP} cannot be generalized to unital positive maps, as it will be discussed in depth in a forthcoming paper, see also \cite{idel2013structure} for the explicit structure of $\Attr(\Phi)$ and $\Phi_{\mathrm{as}}$ in such case. 
\end{Remark}
To sum up, the operator Schwarz inequality is sufficient in order to describe the large-time dynamics of an open quantum system. In other words, complete positivity should be relevant in order to study the \emph{transient} evolution of the system, i.e. how the system approaches the asymptotic states of the dynamics. This still remains a big challenge, as confirmed by the few general works on the topic~\cite{chruscinski2021universal,chruscinski2021constraints}.

Indeed, although completely positive maps are well characterized by a variety of representations, it is still hard to unveil physically transparent information from this property. In particular, as already underlined in Remark~\ref{spectr_UP}, the spectral properties discussed in Subsection~\ref{asymptotic_definition_section} still hold for unital positive maps. Thus we still need to fully understand the spectral implications of complete positivity.

\section{Conclusions}
\label{concl}
In this Article we have explored the large-time dynamics of finite-dimensional open quantum systems in the Heisenberg picture. In particular, following an algebraic approach, we introduced a new space $\Nstar$, called the Choi-Effros decoherence-free algebra of a UCP map $\Phi$, obtained from the standard one, see Eq.~\eqref{dec_def}, by replacing the composition product with the Choi-Effros product~\eqref{star_prod_H}. Such definition was motivated by Theorem~\ref{aut_th_UCP}, which establishes that the appropriate product in the asymptotic limit is the Choi-Effros one. The two most important features of this space, which do not hold for the standard decoherence-free algebra $\mathcal{N}$, are the following:
\begin{enumerate}
\item it has a direct-sum decomposition (Theorem~\ref{sum_dec_N_th}), implying that, in particular, $\Nstar \supseteq \Attr(\Phi)$, the attractor subspace of $\Phi$ [condition \textit{b)} of Theorem~\ref{mul_ch_UCP}];  
\item $\Attr(\Phi)=\Nstar$ is equivalent to the faithfulness of $\Phi$ (Theorem~\ref{faith_mul_ch}).
\end{enumerate}
In particular, the clear relationship between $\Nstar$ and $\Attr(\Phi)$ sheds light on the interplay between the spectral and multiplicative approaches to the asymptotics. As a consequence of Theorem~\ref{sum_dec_N_th}, $\Nstar$ is a ${C}^\ast$-algebra with respect to the composition product $\cdot$, the operator norm $\| \cdot \|$, and the adjoint $\ast$ of $\mathcal{B}(\mathcal{H})$ and, indeed, by definition, $\Nstar$ is the largest ${C}^\ast$-subalgebra of $\mathcal{B}(\mathcal{H})$ on which $(\Phi^n)_{n\in \mathbb{N}}$ acts as a $\ast$-homomorphism with respect to the Choi-Effros product.

Finally, as shown in Section~\ref{Schw_CP}, our results do not depend on the complete positivity of the open-system Heisenberg dynamics $\Phi$, but only on the operator Schwarz inequality~\eqref{op_Schw_ineq}, a quite less strict requirement. 
%
%
Therefore, given the different applications of this property in quantum dynamics, also mentioned in the Introduction, it would be interesting to understand the physical interpretation of the operator Schwarz inequality, beyond its interesting mathematical form and its physical implications.

%
%
%

\bmhead{Acknowledgments}
We acknowledge support from INFN through the project ``QUANTUM'', from the Italian National Group of Mathematical Physics (GNFM-INdAM), and from the Italian funding within the ``Budget MUR - Dipartimenti di Eccellenza 2023--2027''  - Quantum Sensing and Modelling for One-Health (QuaSiModO).\ P.F. acknowledges support from Italian PNRR MUR project CN00000013 -``Italian National Centre on HPC, Big Data and Quantum Computing''.\ D.A. acknowledges support from PNRR MUR project PE0000023-NQSTI.\ A.K. acknowledges the support by the National Science Centre (Poland) through the SONATA BIS project No. 019/34/E/ST2/00369. 

%
%
%
%
%
%

\begin{appendices}
\section{Choi-Effros product and decoherence-free algebra}
\label{Nstar_seminorm}
In this appendix we will discuss again the relation between the Choi-Effros decoherence-free algebra $\Nstar$ and the attractor subspace $\Attr(\Phi)$ of a UCP map $\Phi$. 

Remember that a $B^\ast$-algebra is not generally a $C^\ast$-algebra, as the $C^\ast$-identity is in general violated. Nevertheless, there are situations where the first can be equipped with a $C^\ast$-seminorm. This is exactly what happens for $\Nstar$, which is a $B^\ast$-algebra with respect to the Choi-Effros product $\star$, see Theorem~\ref{C*-dec_th}. 

First, let us recall the definition of $C^\ast$-seminorm~\cite[§ 39 Definition 1]{bonsall2012complete}.
	\begin{definition}
		 A $C^\ast$-seminorm $p$ on a $B^\ast$-algebra $\mathcal{A}$ is a seminorm satisfying for all $a,b\in \mathcal{A}$
		\begin{align}
			&p(a \cdot b) \leqslant p(a)p(b),\label{submul_prop}\\
			& p(a^\ast \cdot a)=p(a)^2. \label{Cstar_prop}
		\end{align}
		A non-degenerate $C^\ast$-seminorm is called a $C^\ast$-norm. 
	\end{definition}
	According to Theorem~\ref{C*-dec_th}, $(\Nstar , \star , \| \cdot \|)$ is a $B^\ast$-algebra.
	Now, we will prove that such a structure cannot be promoted to a $C^\ast$-algebra. 
	
	Indeed, motivated by~\cite[Corollary 2.2.6]{bratteli2012operator}, 
	let us define the following functional on $\Nstar$
	\begin{equation}
\label{spec_seminorm}
\omega(X)=\rho(X^\ast \star X)^{\frac{1}{2}},
\end{equation}
with $\rho(Y)$ denoting the spectral radius of $Y\in\mathcal{B}(\mathcal{H})$. In the following lemma we will prove that the functional~\eqref{spec_seminorm} is a $C^\ast$-seminorm on $\Nstar$.
\begin{Lemma}
Let $\Phi$ be a UCP map with Choi-Effros decoherence-free algebra $\Nstar$. Then the functional $\omega$ on $\Nstar$ defined by~\eqref{spec_seminorm} is a $C^\ast$-seminorm. 
\end{Lemma}
\begin{proof}
Non-negativity and homogeneity of $\omega$ are easy to verify. Now, let us prove that $\omega$ satisfies the triangle inequality, it is submultiplicative [property~\eqref{submul_prop}], and it satisfies the $C^\ast$-identity~\eqref{Cstar_prop}.

\textit{Triangle inequality:} The proof is based on 
\begin{equation}
\label{aut_P_p}
X\in \Nstar \Rightarrow X^\ast \star X = \mathcal{P}_{\mathrm{P}}(X^\ast \star X)=\mathcal{P}_{\mathrm{P}}(X)^\ast\star \mathcal{P}_{\mathrm{P}}(X).
\end{equation}
Moreover, as a consequence of the uniqueness of the ${C}^\ast$-norm for a $\ast$-algebra, and according to Theorem~\ref{C_star_UCP}, we obtain
\begin{equation}
\label{C*_norm_Attr}
\omega(X)=\| X \|, \quad X\in\Attr(\Phi).
\end{equation}
Thus, given $X,Y\in \Nstar$, we have
\begin{equation}
\begin{split}
 \omega(X+Y)&=\rho((X+Y)^\ast \star (X+Y))^{\frac{1}{2}} = \rho(\mathcal{P}_{\mathrm{P}}(X+Y)^\ast \star \mathcal{P}_{\mathrm{P}}(X+Y))^{\frac{1}{2}} \\&= \omega(\mathcal{P}_{\mathrm{P}}(X+Y)) =  \| \mathcal{P}_{\mathrm{P}}(X+Y) \| \leqslant \| \mathcal{P}_{\mathrm{P}}(X) \| + \| \mathcal{P}_{\mathrm{P}}(Y) \| \\&= \omega(\mathcal{P}_{\mathrm{P}}(X)) + \omega(\mathcal{P}_{\mathrm{P}}(Y)) \\&= \rho(\mathcal{P}_{\mathrm{P}}(X)^\ast \star \mathcal{P}_{\mathrm{P}}(X))^{\frac{1}{2}} + \rho(\mathcal{P}_{\mathrm{P}}(Y)^\ast \star \mathcal{P}_{\mathrm{P}}(Y))^{\frac{1}{2}} \\&
 = \rho(X^\ast \star X)^{\frac{1}{2}} + \rho(Y^\ast \star Y)^{\frac{1}{2}} =  \omega(X) + \omega(Y). 
\end{split}
\end{equation}

\textit{Submultiplicativity:} 
Given $X,Y \in \Nstar$, one has 
\begin{equation}
\begin{split}
\omega(X \star Y) &=  \| X \star Y \| \leqslant \| \mathcal{P}_{\mathrm{P}}(X) \| \| \mathcal{P}_{\mathrm{P}}(Y) \| \\&=\rho(\mathcal{P}_{\mathrm{P}}(X)^\ast \star \mathcal{P}_{\mathrm{P}}(X))^{\frac{1}{2}}\rho(\mathcal{P}_{\mathrm{P}}(Y)^\ast \star \mathcal{P}_{\mathrm{P}}(Y))^{\frac{1}{2}}\\&=\omega(X)\omega(Y),
\end{split}
\end{equation}
where we used again~\eqref{aut_P_p} and~\eqref{C*_norm_Attr}.

\textit{$C^\ast$-identity:}
Given $X \in \Nstar$, one obtains
\begin{equation}
\omega(X) = \rho(\mathcal{P}_{\mathrm{P}}(X)^\ast \star \mathcal{P}_{\mathrm{P}}(X))^{1/2} = \omega(\mathcal{P}_{\mathrm{P}}(X)) = \| \mathcal{P}_{\mathrm{P}}(X) \|,
\end{equation}
and, analogously, 
\begin{equation}
\omega(X^\ast \star X)^{1/2} =  \omega(\mathcal{P}_{\mathrm{P}}(X)^\ast \star \mathcal{P}_{\mathrm{P}}(X))^{1/2} = \| \mathcal{P}_{\mathrm{P}}(X)^\ast \star \mathcal{P}_{\mathrm{P}}(X) \|^{1/2}.
\end{equation}
Thus, property~\eqref{Cstar_prop} follows from the $C^\ast$-identity within $(\Attr(\Phi) , \star , \| \cdot \|)$.
\end{proof}
Unfortunately, $\omega$ is degenerate since
\begin{equation}
\omega(X)=\rho(X^\ast \star X)^{\frac{1}{2}}=\rho(\mathcal{P}_{\mathrm{P}}(X^\ast X))^{\frac{1}{2}}=0 \;\Leftrightarrow \;\mathcal{P}_{\mathrm{P}}(X^\ast X)=0,
\end{equation}
which does not imply that $X=0$ in general, see Example~\ref{ex_1}. However, it is possible to re-obtain, up to a $\ast$-automorphism, the attractor subspace from the algebra $\Nstar$ as a quotient space, i.e.\ to have an alternative proof of Theorem~\ref{sum_dec_N_th} using the formalism of $C^\ast$-seminorms.

\begin{proof}[Second proof of Theorem~\ref{sum_dec_N_th}]
Given the $B^\ast$-algebra $\Nstar$, it is standard to obtain a $C^\ast$-algebra by taking the quotient~\cite[§ 39 Corollary 4]{bonsall2012complete}
\begin{equation}
	\label{quotient_star}
\mathcal{Q}_{\star}:= \Nstar / \Ker(\omega)
\end{equation}
with respect to the kernel of $\omega$
\begin{equation}
\Ker(\omega)=\{ X \in \Nstar \,,\, \mathcal{P}_{\mathrm{P}}(X^\ast X)=0 \},
\end{equation}
which is a $\ast$-ideal of $\Nstar$. Also, it is easy to check that
\begin{equation}
\Ker(\omega)=\mathcal{K}({\mathcal{P}_{\mathrm{P}}}),
\end{equation}
with $\mathcal{K}({\mathcal{P}_{\mathrm{P}}})$ given by Eq.~\eqref{eq:definition_K_idempotent_channel}.
The inclusion $\mathcal{K}({\mathcal{P}_{\mathrm{P}}}) \subseteq \Ker(\omega)$ is equivalent to $\mathcal{K}({\mathcal{P}_{\mathrm{P}}})\subseteq \Nstar$, which was already proved in Eqs.~\eqref{aut_K_P} and~\eqref{n_Schwarz_arg}. 
Conversely, take $X$ in $\Ker(\omega)$ and observe that from the operator Schwarz inequality~\eqref{op_Schw_ineq} for the peripheral projection $\mathcal{P}_{\mathrm{P}}$
\begin{equation}
0=\mathcal{P}_{\mathrm{P}}(X^\ast X)\geqslant \mathcal{P}_{\mathrm{P}}(X)^\ast \mathcal{P}_{\mathrm{P}}(X) \Rightarrow \mathcal{P}_{\mathrm{P}}(X)^\ast \mathcal{P}_{\mathrm{P}}(X)=0,
\end{equation}
implying $\mathcal{P}_{\mathrm{P}}(X)=0$. Thus, since $X\in\Nstar$, we have
\begin{equation}
	\mathcal P_{\mathrm{P}}(XX^\ast)=\mathcal P_{\mathrm{P}}(X)\star\mathcal P_{\mathrm{P}}(X)^\ast=0,
\end{equation} 
so $X\in\mathcal K(\mathcal P_{\mathrm{P}})$. As a result, the quotient space $\mathcal{Q}_{\star}$ coincides (as a vector space) to the quotient $\mathcal Q$ given by Eq.~\eqref{Q_def}.

A product can naturally be introduced on $\mathcal{Q}_{\star}$ as
\begin{equation}
\label{Choi_quot_prod}
	\lfloor X \rfloor \star \lfloor Y  \rfloor = \{ X \star Y + K\,,\, K \in \Ker(\omega) \},
\end{equation}

as well as a norm $\| \cdot \|_{\star}$ defined via
\begin{equation}
	\label{norm_star_def}
	\norm{\lfloor X \rfloor}_\star = \omega(Z),\quad Z  \in \lfloor X \rfloor.
\end{equation}
Note that this norm is well-defined as, given $X\in \Nstar$ and $K\in\mathcal K({\mathcal P_{\mathrm{P}}})$, we have
\begin{equation}
\begin{split}
\omega(X+K)^2&=\rho(\mathcal P_{\mathrm{P}}((X+K)^\ast(X+K))) \\&=\rho(\mathcal P_{\mathrm{P}}(X^\ast X)+\mathcal P_{\mathrm{P}}(X^\ast K)+\mathcal P_{\mathrm{P}}(K^\ast X)+\mathcal P_{\mathrm{P}}(K^\ast K))
 \\ &= \rho(\mathcal P_{\mathrm{P}}(X^\ast X))=\omega(X)^2.
\end{split}
\end{equation}
In the third equality, we used the fact that $\mathcal{K}({\mathcal{P}_{\mathrm{P}}})$ is a $\ast$-ideal in $\Nstar$.
Since $\omega$ and, consequently, $\| \cdot \|_{\star}$ satisfies the ${C}^\ast$-identity, we can conclude that $(\mathcal{Q}_{\star} , \star , \| \cdot \|_{\star})$
is a ${C}^\ast$-algebra, $\ast$-automorphic to $\Attr(\Phi)$ via the map~\eqref{iso_Attr}.

Let us now prove that the quotient spaces $\mathcal{Q}$ and $\mathcal{Q}_{\star}$ coincide as $C^\ast$-algebras.
Specifically, as it is clear from~\eqref{phi_auto},
\begin{equation}
\lfloor X \rfloor \star \lfloor Y  \rfloor = \lfloor X \rfloor \cdot \lfloor Y  \rfloor, \quad X,Y \in \Attr(\Phi),
\end{equation}
Indeed, it holds for any pair $X,Y \in \Nstar$ since
\begin{equation}
\label{cl_attr}
 \lfloor X \rfloor  =  \{ X+K , K \in \mathcal{K}({\mathcal{P}_{\mathrm{P}}}) \} =  \{  \mathcal{P}_{\mathrm{P}}(X)+K ,K \in \mathcal{K}({\mathcal{P}_{\mathrm{P}}}) \} =  \lfloor \mathcal{P}_{\mathrm{P}}(X) \rfloor .
\end{equation}
Therefore, by the uniqueness of the norm making a $\ast$-algebra a ${C}^\ast$-algebra~\cite[Corollary 2.2.6]{bratteli2012operator}, the norms $\| \cdot \|$ and $\| \cdot \|_{\star}$ defined by~\eqref{norm_Q_def} and~\eqref{norm_star_def}, respectively, coincide. This concludes the alternative proof 
of Theorem~\ref{sum_dec_N_th}.
\end{proof}
\end{appendices}
\section*{Declarations}
\begin{itemize}
\item \textbf{Conflict of interest:} On behalf of all authors, the corresponding author states that there is no Conflict of interest.
\item \textbf{Data availability:} This manuscript has no associated data.
\end{itemize}






\begin{thebibliography}{aa}
	
	\bibitem{albert2019asymptotics}
	V.~V. Albert:
	\newblock Asymptotics of quantum channels: conserved quantities, an adiabatic
	limit, and matrix product states.
	\newblock {Quantum} \textbf{3}, 151 (2019)
	
	\bibitem{contraction2}
	S. Alipour, D.~Chru{\'s}ci{\'n}ski, P. Facchi, G. Marmo, S. Pascazio,  A. T.  Rezakhani:
	\newblock {Dynamical algebra of observables in dissipative quantum systems}.
	\newblock { J. Phys. A: Math. Theor.} \textbf{49}, 065301 (2017)
	
	\bibitem{AF_bounds}
	D.~Amato, P.~Facchi:
	\newblock {Number of steady states of quantum evolutions}.
	\newblock{Sci. Rep.} \textbf{14}, 14366 (2024)
	
	\bibitem{AFK_asympt_2}
	D.~Amato, P.~Facchi, A.~Konderak:
	\newblock {Asymptotics of quantum channels}.
	\newblock { J. Phys. A: Math. Theor.} \textbf{56}, 265304 (2023)

    \bibitem{AFK_asympt}
	D.~Amato, P.~Facchi, A.~Konderak:
	\newblock {Asymptotic Dynamics of Open Quantum Systems and Modular Theory}.
	\newblock In ``Quantum Mathematics II'', edited by M. Correggi and M. Falconi,
	Springer INdAM Series Vol. 58, p. 169––181 (Springer, Singapore, 2023)

    \bibitem{AFK_asympt_3}
	D.~Amato, P.~Facchi, A.~Konderak:
	\newblock {Attractor Subspace and Decoherence-Free Algebra of Quantum Dynamics}.
	\newblock In ``Singularities, Asymptotics, and Limiting Models'', edited by B. Cassano, F.D. Cunden, M. Gallone, M. Ligab{\`o} and A. Michelangeli,
	Springer INdAM Series Vol. 64, p. 79-96  (Springer, Singapore, 2025)

    \bibitem{AFK_asympt_4}
	D.~Amato, P.~Facchi, A.~Konderak:
	\newblock {Asymptotic Dynamics in the Heisenberg Picture: Attractor Subspace and Choi-Effros Product}.
	\newblock {Open Syst. Inf. Dyn.} \textbf{32}, 2550014 (2025)

    \bibitem{AFM_2025}
	D.~Amato, P.~Facchi, G.~Marmo:
	\newblock {Bridging Classical and Quantum Worlds: Maps,
States, and Evolutions}.
	\newblock {accepted for publication in Open Syst. Inf. Dyn.}, arXiv:2511.09390 [quant-ph] (2025)
    

	
	\bibitem{Asorey2008}
	M.~Asorey, A.~Kossakowski, G.~Marmo, E.~C.~G. Sudarshan:
	\newblock {Unital Positive Maps and Quantum States}.
	\newblock {Open Syst. Inf. Dyn.} \textbf{15}, 123--134 (2008)
	
%
	
	\bibitem{batkai2012decomposition}
	A.~B{\'a}tkai, U.~Groh, D.~Kunszenti-Kov{\'a}cs, M.~Schreiber:
	\newblock {Decomposition of operator semigroups on W*-algebras}.
	\newblock {Semigr. Forum} \textbf{84}, 8--24 (2012)
	
	\bibitem{rajaramaperipheral_2}
	B.~V.~R.~Bhat, S.~Kar, B.~Talwar:
	\newblock {Peripherally automorphic unital completely positive maps.}
	\newblock {Linear Algebra Its Appl.} \textbf{678}, 191-205 (2023)
	
	\bibitem{rajarama2022peripheral}
	B.~V.~R.~Bhat, B.~Talwar, S.~Kar:
	\newblock {Peripheral Poisson boundary}.
	\newblock {Isr. J. Math.} (2025)
	
	\bibitem{blanchard2003quantum}
	Ph.~Blanchard, P.~Lugiewicz, R.~Olkiewicz:
	\newblock {From quantum to quantum via decoherence}.
	\newblock { Phys. Lett. A} \textbf{314}, 29--36 (2003)
	
	\bibitem{blanchard2003decoherence}
	Ph.~Blanchard, R.~Olkiewicz:
	\newblock {Decoherence induced transition from quantum to classical dynamics}.
	\newblock { Rev. Math. Phys.} \textbf{15}, 217--243 (2003)
	
	\bibitem{blanchard2006decoherence}
	Ph.~Blanchard, R.~Olkiewicz:
         \newblock {Decoherence as Irreversible Dynamical Process in Open Quantum Systems}.
	\newblock{In ``Open Quantum Systems '', edited by S.~Attal, A.~Joye, C.~A.~Pillet},
	\newblock{Lecture Notes in Mathematics Vol. 1882, p. 117 (Springer, Berlin, Heidelberg 2006)}
	
	\bibitem{bonsall2012complete}
	F.~F. Bonsall, J.~Duncan:
	\newblock {Complete Normed Algebras}.
	\newblock Springer Science \& Business Media, New York (2009)
	
	\bibitem{bratteli2012operator}
	O.~Bratteli, D.~W.~Robinson.
	\newblock { {Operator Algebras and Quantum Statistical Mechanics: Volume 1:
			C*- and W*-Algebras. Symmetry Groups. Decomposition of States}}.
	\newblock Springer Science \& Business Media, New York (2012)
	
	\bibitem{bulinskii1995some}
	A.~V.~Bulinskii:
	\newblock {Some Asymptotic Properties of W*-Dynamical Systems}.
	\newblock {Funct. Anal. its Appl.} \textbf{29}, 123--126 (1995)
	
	\bibitem{carbone2020period}
	R.~Carbone, A.~Jen{\v{c}}ov{\'a}:
	\newblock {On Period, Cycles and Fixed Points of a Quantum Channel}.
	\newblock {Ann. Henri Poincar{\'e}} \textbf{21}, 155--188 (2020)
	
	\bibitem{carbone2013Markov}
	R.~Carbone, E. Sasso, V. Umanit{\'a}: 
	\newblock{Decoherence for Quantum Markov Semigroups on Matrix Algebras}.
	\newblock{Ann. Henri Poincar{\'e}} \textbf{14}, 681--697 (2013)
	
	\bibitem{carlen2022characterizing}
	E.~A. Carlen, A.~M{\"u}ller-Hermes:
	\newblock {Characterizing Schwarz maps by tracial inequalities}.
	\newblock {Lett. Math. Phys.} \textbf{113}, 17--34 (2023)
	
	\bibitem{carlen2022monotonicity}
	E.~A. Carlen, H.~Zhang:
	\newblock {Monotonicity versions of Epstein's Concavity Theorem and related
		inequalities}.
	\newblock {Linear Algebra Appl.} \textbf{654}, 289--310 (2022)
	
	\bibitem{choi1974schwarz}
	M.-D. Choi:
	\newblock {A {S}chwarz inequality for positive linear maps on C*-algebras}.
	\newblock {Ill. J. Math} \textbf{18}, 565--574 (1974)
	
	\bibitem{choi1980some}
	M.-D.~Choi:
	\newblock {Some assorted inequalities for positive linear maps on
		${C}^\ast$-algebras}.
	\newblock {J. Oper. Theory} \textbf{4}, 271--285 (1980)
	
	\bibitem{choi1977injectivity}
	M.-D.~Choi, E.~G.~Effros:
	\newblock {Injectivity and Operator Spaces}.
	\newblock { J. Funct. Anal.} \textbf{24}, 156--209 (1977)
	
	\bibitem{choi2009multiplicative}
	M.-D.~Choi, N.~Johnston, D.~W.~Kribs:
         \newblock {The multiplicative domain in quantum error correction}.
	\newblock { J. Phys. A: Math. Theor.} \textbf{42},  245303 (2009)
	
	\bibitem{contraction}
	D.~Chru{\'s}ci{\'n}ski, P.~Facchi, G.~Marmo, S.~Pascazio:
	\newblock {The Observables of a Dissipative Quantum System}.
	\newblock {Open Syst. Inf. Dyn.} \textbf{19}, 1250002 (2012)
	
	\bibitem{chruscinski2021constraints}
	D.~Chru{\'s}ci{\'n}ski, R.~Fujii, G.~Kimura, H.~Ohno:
	\newblock {Constraints for the spectra of generators of quantum dynamical
		semigroups}.
	\newblock { Linear Algebra Appl.} \textbf{630}, 293--305 (2021)
	
	\bibitem{chruscinski2021universal}
	D.~Chru{\'s}ci{\'n}ski, G.~Kimura, A.~Kossakowski, Y.~Shishido:
	\newblock {Universal Constraint for Relaxation Rates for Quantum Dynamical
		Semigroup}.
	\newblock {Phys. Rev. Lett.} \textbf{127}, 050401 (2021)
	
	
	\bibitem{chruscinski2020kadison}
	D.~Chru{\'s}ci{\'n}ski, F.~Mukhamedov,  M.~A.~Hajji:
	\newblock On {K}adison-{S}chwarz {A}pproximation to {P}ositive {M}aps.
	\newblock {Open Syst. Inf. Dyn.} \textbf{27}, 2050016 (2020)
	
	\bibitem{chruscinski2022dynamical}
	{D. Chru{\'s}ci{\'n}ski:}
	\newblock {Dynamical maps beyond Markovian regime}.
	\newblock {Phys. Rep.} \textbf{992}, 1--85 (2022)
	
	\bibitem{groupoidI}
	F.~M.~Ciaglia, A.~Ibort, G.~Marmo:
	\newblock {Schwinger’s picture of quantum mechanics {I}: Groupoids}.
	\newblock {Int. J. Geom. Methods Mod. Phys.} \textbf{16}, 1950119 (2019)
	
	\bibitem{groupoidII}
	F.~M.~Ciaglia, A.~Ibort, G.~Marmo:
	\newblock {Schwinger’s picture of quantum mechanics {II}: Algebras and
		observables}.
	\newblock {Int. J. Geom. Methods Mod. Phys.} \textbf{16}, 1950136 (2019)
	
		\bibitem{groupoidIII}
	F.~M.~Ciaglia, F.~Di Cosmo, P.~Facchi, A.~Ibort, A.~Konderak, G.~Marmo:
	\newblock {Groupoid and algebra of the infinite quantum spin chain}.
	\newblock {J. Geom. Phys.} \textbf{191}, 104901 (2023)
	
	\bibitem{conway2019course}
	J.~B.~Conway:
	\newblock {{A {C}ourse in {F}unctional {A}nalysis. {S}econd {E}dition}}.
	\newblock Springer-Verlag, New York (1990)
	
	\bibitem{Davidson1996}
	K.~R.~Davidson:
	\newblock {{C*-Algebras by Example}}.
	\newblock American Mathematical Society, Providence (1996)
	
	\bibitem{deleeuw1959almost}
	K.~DeLeeuw, I.~Glicksberg:
	\newblock {Almost periodic compactifications}.
	\newblock {Bull. Am. Math. Soc.} \textbf{65}, 134--139 (1959)
	
	\bibitem{deleeuw1961applications}
	K.~DeLeeuw, I.~Glicksberg:
	\newblock {Applications of almost periodic compactifications}.
	\newblock {Acta Math.} \textbf{105}, 63--97 (1961)
	
	\bibitem{dhahri2010decoherence}
	A.~Dhahri, F.~Fagnola,  R.~Rebolledo:
	\newblock {The decoherence-free subalgebra of a quantum Markov semigroup with
		unbounded generator}.
	\newblock{Infin. Dimens. Anal. Quantum Probab. Relat. Top.} \textbf{13}, 413--433 (2010)
	
	
	\bibitem{evans1977irreducible}
	D.~E.~Evans:
	\newblock {Irreducible Quantum Dynamical Semigroups}.
	\newblock {Commun. Math. Phys.} \textbf{54}, 293--297 (1977)
	
	\bibitem{fgk}
	P.~Facchi, G.~Gramegna, A.~Konderak:
	\newblock {Entropy of Quantum States}.
	\newblock {Entropy} \textbf{23}, 645 (2021)
	
	\bibitem{fagnola_2001}
	F.~Fagnola and R.~Rebolledo:
	\newblock {On the existence of stationary states for quantum dynamical
		semigroups}.
	\newblock {J. Math. Phys.} \textbf{42}, 1296 (2001)
	
	\bibitem{fagnola2002subharmonic}
	F.~Fagnola and R.~Rebolledo:
	\newblock {Subharmonic projections for a quantum {M}arkov semigroup}.
	\newblock {J. Math. Phys.} \textbf{43}, 1074--1082 (2002)
	
	\bibitem{fagnola2008algebraic}
	F.~Fagnola and R.~Rebolledo:
	\newblock {Algebraic conditions for convergence of a quantum Markov semigroup
		to a steady state}.
	\newblock {Infin. Dimens. Anal. Quantum Probab. Relat. Top.} \textbf{11}, 467--474 (2008)
	
	\bibitem{fagnola2017structure}
	F.~Fagnola, E.~Sasso, V.~Umanit{\`a}:
	\newblock {Structure of uniformly continuous quantum Markov semigroups with
		atomic decoherence-free subalgebra}.
	\newblock {Open Syst. Inf. Dyn.} \textbf{24}, 1740005 (2017)
	
	\bibitem{fagnola2019role}
	F.~Fagnola, E.~Sasso, V.~Umanit{\`a}:
	\newblock {The role of the atomic decoherence-free subalgebra in the study of
		quantum Markov semigroups}.
	\newblock {J. Math. Phys.} \textbf{60}, 072703 (2019)
	
	\bibitem{fidaleo2022spectral}
	F.~Fidaleo, F.~Ottomano, S.~Rossi:
	\newblock {Spectral and ergodic properties of completely positive maps and
		decoherence}.
	\newblock {Linear Algebra Appl.} \textbf{633}, 104--126 (2022)
	
	\bibitem{Frigerio_78}
	A.~Frigerio:
	\newblock {{S}tationary {S}tates of {Q}uantum {D}ynamical {S}emigroups}.
	\newblock{Commun. Math. Phys.} \textbf{63}, 269--276 (1978)
	  
	\bibitem{Frigerio_Verri_82}
	A.~Frigerio, M.~Verri:
	\newblock {Long-\uppercase{T}ime \uppercase{A}symptotic \uppercase{P}roperties
		of \uppercase{D}ynamical \uppercase{S}emigroups on ${W}^\ast$-algebras}.
	\newblock {Math Z.} \textbf{180}, 275--286 (1982)
	
	\bibitem{garcia2014non}
	M.~C. Garc{\'\i}a, {\'A}.~R.~Palacios:
	\newblock {{Non-Associative Normed Algebras: Volume 1, The Vidav-Palmer
			and Gelfand-Naimark Theorems}}.
	\newblock Cambridge University Press, Cambridge (2014)
	
	\bibitem{gelfandNormierte}
	I.~Gel'fand:
	\newblock {Normierte Ringe}.
	\newblock {Mat. Sb.} \textbf{9}, 3--24 (1941)
	
	\bibitem{gelfand1943imbedding}
	I.~Gel'fand, M.~Neumark:
	\newblock {On the imbedding of normed rings into the ring of operators in
		{H}ilbert space}.
	\newblock {Mat. Sb.} \textbf{12}, 197--217 (1943)
	
	\bibitem{GKS_76}
	V.~Gorini, A.~Kossakowski, E.~C.~G.~Sudarshan:
	\newblock {Completely positive dynamical semigroups of {N}-level systems}.
	\newblock {J. Math. Phys.} \textbf{17}, 821--825 (1976)
	
	
	\bibitem{hamana1979injective}
	M.~Hamana:
	\newblock {Injective envelopes of ${C}^\ast$-algebras}.
	\newblock {J. Math. Soc. Japan} \textbf{31}, 181--197 (1979)
	
	\bibitem{idel2013structure}
	M.~Idel:
	\newblock {On the structure of positive maps}.
	\newblock {Master's thesis, Technical University of Munich}, (2013)
	
	\bibitem{jacobs1957fastperiodizitatseigenschaften}
	K.~Jacobs:
	\newblock {Fastperiodizit{\"a}tseigenschaften allgemeiner Halbgruppen in
		Banach-R{\"a}umen}.
	\newblock {Math. Z.} \textbf{67}, 83--92 (1957)
	
	\bibitem{jenvcova2012reversibility}
	A.~Jen{\v{c}}ov{\'a}:
	\newblock {Reversibility conditions for quantum operations}.
	\newblock {Rev. Math. Phys.} \textbf{24}, 1250016 (2012)
	
	\bibitem{kato2013perturbation}
	T.~Kato:
	\newblock {{Perturbation Theory for Linear Operators}}.
	\newblock Springer Science \& Business Media, New York (2013)
	
	\bibitem{kielanowicz2017spectral}
	K.~Kielanowicz and A.~{\L}uczak:
	\newblock Spectral properties of {M}arkov semigroups in von {N}eumann algebras.
	\newblock {J. Math. Anal. Appl.} \textbf{453}, 821--840 (2017)
	
	\bibitem{kossakowski1972quantum}
	A.~Kossakowski:
	\newblock {On quantum statistical mechanics of non-Hamiltonian systems}.
	\newblock {Rep. Math. Phys.} \textbf{3}, 247--274 (1972)
	
	\bibitem{lindblad1975completely}
	G.~Lindblad:
	\newblock {Completely Positive Maps and Entropy Inequalities}.
	\newblock {Commun. Math. Phys.}  \textbf{40}, 147--151 (1975)
	
	\bibitem{Lindblad_76}
	G.~Lindblad:
	\newblock {On the Generators of Quantum Dynamical Semigroups}.
	\newblock {Commun. Math. Phys.} \textbf{48}, 119--130 (1976)
	
	\bibitem{nielsen2002quantum}
	M.~A.~Nielsen, I.~Chuang:
	\newblock {Quantum {C}omputation and {Q}uantum {I}nformation}.
	\newblock Cambridge {U}niversity {P}ress, Cambridge (2002)
	
	\bibitem{jex_st_2012}
	J.~Novotn{\'{y}}, G.~Alber, I.~Jex:
	\newblock {Asymptotic properties of quantum Markov chains}.
	\newblock {J. Phys. A: Math. Theor.} \textbf{45}, 485301 (2012)
	
	\bibitem{jex_st_2018}
	J.~Novotn{\'{y}}, J.~Mary{\v{s}}ka, I.~Jex:
	\newblock {Quantum Markov processes: From attractor structure to explicit forms
		of asymptotic states}.
	\newblock {Eur. Phys. J. Plus} \textbf{133}, 310 (2018)
	
	\bibitem{osaka1991positive}
	H.~Osaka:
	\newblock Positive {P}rojections on ${C}^\ast$-{A}lgebras.
	\newblock {Tokyo J. Math.} \textbf{14}, 73--83 (1991)
	
	\bibitem{perez2006matrix}
	D.~Perez-Garcia, F.~Verstraete, M.~M.~Wolf,  J.~I.~Cirac:
	\newblock Matrix product state representations.
	\newblock {Quantum Inf. Comput.} \textbf{7}, 401--430 (2007)
	
	\bibitem{Zoller_res_eng}
	J.~F.~Poyatos, J.~I.~Cirac, P.~Zoller:
	\newblock Quantum {R}eservoir {E}ngineering with {L}aser {C}ooled {T}rapped
	{I}ons.
	\newblock {Phys. Rev. Lett.} \textbf{77}, 4728 (1996)
	
	\bibitem{rahaman2017multiplicative}
	M.~Rahaman:
	\newblock {Multiplicative properties of quantum channels}.
	\newblock {J. Phys. A: Math. Theor.} \textbf{50}, 345302 (2017)
	

	\bibitem{robinson1982strongly}
	D.~W. Robinson:
	\newblock {Strongly Positive Semigroups and Faithful Invariant States}.
	\newblock {Commun. Math. Phys.} \textbf{85}, 129--142 (1982)
	
	\bibitem{Schwinger_2018}
	J.~Schwinger:
	\newblock {{Quantum Kinematics and Dynamics}}.
	\newblock Westview Press, Boulder (1991)
	
		
	\bibitem{siudzinska2021interpolating}
	K.~Siudzi{\'n}ska, S.~Chakraborty, D.~Chru{\'s}ci{\'n}ski:
	\newblock {Interpolating between Positive and Completely Positive Maps: A New
		Hierarchy of Entangled States}.
	\newblock {Entropy} \textbf{23}, 625 (2021)
	
	\bibitem{Spohn_77}
	H.~Spohn:
	\newblock {An algebraic condition for the approach to equilibrium of an open
		{N}-level system}.
	\newblock {Lett. Math. Phys.} \textbf{2}, 33--38 (1977)
	
	\bibitem{stormer2007multiplicative}
	E.~St{\o}rmer:
	\newblock {Multiplicative properties of positive maps}.
	\newblock {Math. Scand.} \textbf{100}, 184--192 (2007)
	
	\bibitem{stormer2012positive}
	E.~St{\o}rmer:
	\newblock {{Positive Linear Maps of Operator Algebras}}.
	\newblock Springer Science \& Business Media, New York (2012)
	
	\bibitem{Wolf_res_eng}
	F.~Verstraete, M.~M.~Wolf, J.~I.~Cirac:
	\newblock {Quantum computation and quantum-state engineering driven by
		dissipation}.
	\newblock {Nat. Phys.} \textbf{5}, 633--636 (2009)
	
	\bibitem{John_God_Neumann_on_the}
	J.~von~Neumann: \newblock On an algebraic generalization of the quantum mechanical formalism
	(part I).
	\newblock {Mat. Sb.} \textbf{1}, 415--484 (1936)
	
	\bibitem{wolf2012quantum}
	M.~M.~Wolf.
	\newblock Quantum {C}hannels \& {O}perations: {G}uided {T}our,
	\newblock Online {L}ecture {N}otes, (2012)
	
	\bibitem{wolf2010inverse}
	M.~M. Wolf, D.~Perez-Garcia:
	\newblock {The Inverse Eigenvalue Problem for Quantum Channels}.
	\newblock arXiv:1005.4545 [quant-ph] (2010)
	
	\bibitem{zanardi2014coherent}
	P.~Zanardi, L.~C.~Venuti:
	\newblock {Coherent Quantum Dynamics in Steady-State Manifolds of Strongly
		Dissipative Systems}.
	\newblock {Phys. Rev. Lett.} \textbf{113}, 240406 (2014)
	
\end{thebibliography}



\end{document}